\documentclass[aps,pra,twocolumn,nofootinbib]{revtex4-1} 

\usepackage[utf8,latin1]{inputenc}   	  				 
\usepackage[T1]{fontenc}          	    			     
\usepackage[british]{babel}        					     
\usepackage{mathpazo}  			   	      				 
\usepackage[dvipsnames,x11names]{xcolor} 				 
\usepackage[babel]{microtype}                            
\usepackage{amsmath,amssymb,amsthm,bm,amsfonts,bbm}      
\usepackage{mathtools}									 
\usepackage[colorlinks=true,citecolor=green,linkcolor=Purple,urlcolor=Magenta]{hyperref}     

\newcommand{\ket}[1]{\vert#1\rangle}
\newcommand{\bra}[1]{\langle#1\vert}

\newcommand{\ketbra}[2]{\vert #1 \rangle \langle #2 \vert}

\newcommand{\tr}{\text{\normalfont Tr}}

\newcommand{\ie}{\textit{i.e.}}
\newcommand{\eg}{\textit{e.g.}}

\DeclarePairedDelimiter{\ceil}{\lceil}{\rceil}
\DeclarePairedDelimiter{\floor}{\lfloor}{\rfloor}

\makeatletter 
\newcommand{\doublewidetilde}[1]{{%
  \mathpalette\double@widetilde{#1}%
}}
\newcommand{\double@widetilde}[2]{%
  \sbox\z@{$\m@th#1\widetilde{#2}$}%
  \ht\z@=.9\ht\z@
  \widetilde{\box\z@}%
}

\newcommand{\map}[1]{\widetilde{#1}}  
\newcommand{\smap}[1]{\doublewidetilde{{#1}_{}}}
\makeatother

\renewcommand{\H}{\mathcal{H}}
\newcommand{\I}{\mathcal{I}}

\renewcommand{\O}{\mathcal{O}}
\newcommand{\C}{\mathfrak{C}}
\newcommand{\im}{\mathrm{i}}

\newcommand{\blue}[1]{\textcolor{blue}{#1}}

\newtheorem{theorem}{Theorem}
\newtheorem{lemma}{Lemma}

\newcommand{\todai}{Department of Physics, Graduate School of Science, The University of Tokyo, Hongo 7-3-1, Bunkyo-ku, Tokyo 113-0033, Japan}

\begin{document}

\title{Probabilistic exact universal quantum circuits for transforming unitary operations} 

\author{ Marco Túlio Quintino }
\affiliation{\todai}

\author{   Qingxiuxiong Dong }
\affiliation{\todai}

\author{   Atsushi Shimbo}
\affiliation{\todai}

\author{ Akihito Soeda }
\affiliation{\todai}

\author{ Mio Murao}
\affiliation{\todai}

\date{15th of April 2020}


\begin{abstract}
		This paper addresses the problem of designing universal quantum circuits to transform $k$ uses of a $d$-dimensional unitary input-operation into a unitary output-operation in a  probabilistic heralded manner.
		Three classes of {protocols are considered, parallel circuits, where the input-operations can be performed simultaneously, adaptive circuits, where sequential uses of the input-operations are allowed,  and general protocols, where the use of the input-operations may be performed without a definite causal order.} For these three classes, we develop a systematic semidefinite programming approach that finds a circuit which obtains the desired transformation with the maximal success probability.
		We then analyse in detail three particular transformations; unitary transposition, unitary complex conjugation,  and unitary inversion. For unitary transposition and unitary inverse, we prove that for any fixed dimension $d$, adaptive circuits have an exponential improvement in terms of uses $k$ when compared to parallel ones. For unitary complex conjugation and unitary inversion we prove that if the number of uses $k$ is strictly smaller than $d-1$, the probability of success is necessarily zero. We also discuss the advantage of indefinite causal order protocols over causal ones and introduce the concept of delayed input-state quantum circuits.
\end{abstract}



\maketitle


\section{Introduction}
		
		In quantum mechanics, deterministic transformations between states are represented by quantum channels and probabilistic transformations by quantum instruments, which consist of quantum channels followed by a quantum measurement. Understanding the properties of quantum channels and quantum instruments is a standard and well established field of research with direct impact for theoretical and applied quantum physics \cite{krausbook,chuang}. Similarly to states, quantum channels may also be subjected to universal transformation in a paradigm usually referred as \textit{higher order} transformations. Higher order transformations can be formalised by quantum supermaps \cite{chiribella07,chiribella08} and physically implemented by means of quantum circuits (see Fig.\,\ref{fig:1}). Despite its fundamental value and potential for applications (\eg, quantum circuit designing \cite{chiribella07}, quantum process tomography \cite{bisio08}, testing causal hypothesis \cite{chiribella18}, channel discrimination \cite{chiribella08C},  aligning reference frames \cite{bartlett08,bisio10}, analysing the role of causal order \cite{chiribella09,oreshkov11}), higher order transformations are still not well understood when compared to quantum channels and quantum instruments.
		
	  	Reversible operations play an important role in mathematics and in various physical theories such as quantum mechanics and thermodynamics.  In quantum mechanics, reversible operations are represented by unitary operators \cite{wolfbook_coro,cariello12}. 
	This work considers \textit{universal} transformations between reversible quantum transformations, that is, we seek for quantum circuits which implement the desired transformation for any  unitary operation of some fixed dimension without any further specific details of the input unitary operation. From a practical perspective, this universal requirement ensures that the circuit does not require any readjustments or modification when different inputs are considered and the circuit implements the desired transformation even when the description of the $d$-dimensional reversible operation is unknown. Note that the universal requirement also imposes strong constraints on transformations which can be physically realised. {A well-known example which pinpoints these constraints when considering quantum states is quantum cloning, although it is simple to construct a quantum device that clones qubits which are promised to be in the state $\ket{0}$ or $\ket{1}$, it is not possible to design a universal quantum transformation that clones all qubit states \cite{wootters82}. Another interesting example can be found in Ref.\,\cite{buzek99} where the authors consider universal not gates for qubits.}
	 \begin{figure} [t]

\begin{center}
\includegraphics[scale=.18]{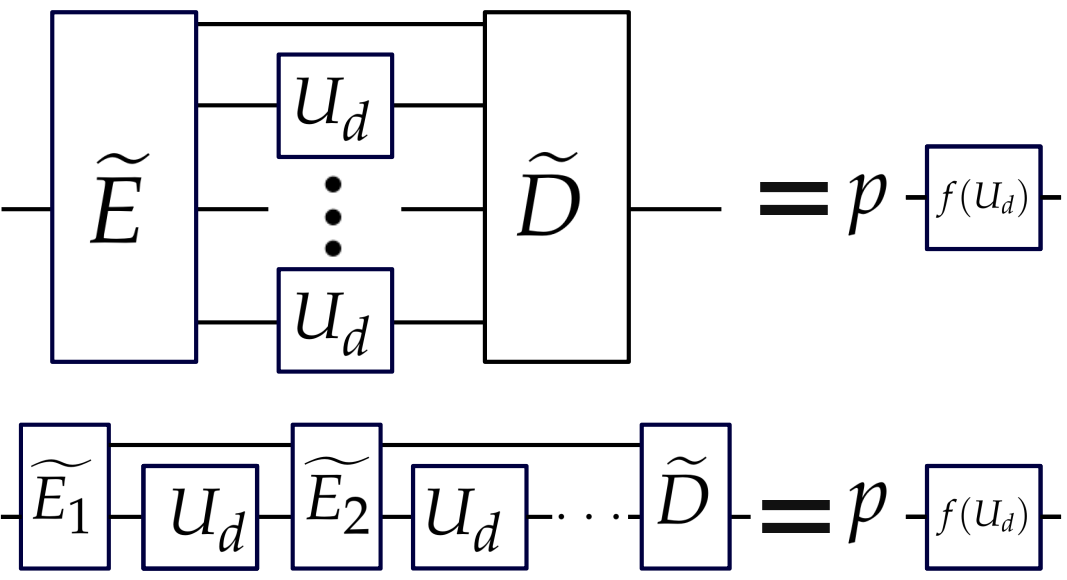} 
\end{center}
\caption{Pictorial representation of parallel and adaptive quantum circuits that transform $k$ uses of a $d$-dimensional arbitrary unitary input-operation described by $U_d$ into another unitary operation described by $f(U_d)$. The circuit elements $\map{E}$ and $\map{E_i}$ are quantum deterministic operations, \ie, quantum channels, that may be interpreted as encoders and the element $\map{D}$ stands for decoder, a probabilistic quantum operation (quantum instrument), that is, a quantum channel with a quantum measurement that when the ``correct''	 outcome is obtained, the target transformation is obtained perfectly. 
}

\label{fig:1}
\end{figure}

	Universal transformations on reversible quantum operations have been studied from several perspectives and motivations such as gate discrimination \cite{acin01,dariano02}, cloning unitary operations \cite{chiribella08-2}, preventing quantum systems to evolve \cite{sardharwalla16,navascues17}, designing quantum circuits \cite{chiribella07}, learning the action of a unitary \cite{vidal01,sasaki02,gammelmark09,bisio10,sedlak18}, transforming unitary operations into their complex conjugate \cite{miyazaki17}, understanding the role of causal order in quantum mechanics \cite{chiribella09,araujo15}, and others \cite{ebler18,salek18}. Probabilistic exact transformations between multiple uses of reversible operations via quantum circuits have been considered in Ref.\,\cite{PRL} where the authors target transforming an arbitrary unitary operation into its inverse and in Ref.\,\cite{bisio13} where the authors consider the case where the unitary input-operation  and the unitary output-operation are two different representations of the same group element. Also, Ref.\,\cite{navascues17,trillo19} consider the probabilistic exact circuits which act only in an auxiliary system which interacts to the target one via some random Hamiltonian.

		 	This paper is focused on designing universal quantum circuits which are not exclusively tailored for a particular class of input-operations. That is, it should attain the desired transformation for any $d$-dimensional unitary operation even if its description is not known. In particular, we focus on  probabilistic heralded transformations between multiple uses of reversible operations. In particular, we focus on  probabilistic heralded transformations between multiple uses of reversible operations. More precisely, we consider circuits which make use of a quantum measurement with an output associated with success and, when the success outcome is obtained, the transformation is implemented perfectly. { We consider three classes of quantum protocols: parallel circuits, where the input-operations can be performed simultaneously, adaptive circuits where the input-operations may be used sequentially, and general protocols which may not be realisable by quantum circuits but are consistent with quantum theory when the use of the input-operations may be performed in an indefinite causal order \cite{chiribella09,oreshkov11,araujo14}. 
		 	We present a systematic approach based on semidefinite programming that allows us to analyse transformations which is linear on quantum operations.  We then analyse in details three particular transformations, unitary transposition, unitary complex conjugation, and unitary inversion.}

		Section\,\ref{sec:intro} reviews results related to quantum circuits such as quantum supermaps, quantum combs, process matrices, and other important concepts.
		Section\,\ref{sec:SDP} presents a general semidefinite programming (SDP) approach to design optimal probabilistic exact quantum circuits.
		Section\,\ref{sec:delayed} introduces the concept of delayed input-states quantum circuits.  
		Section\,\ref{sec:conjugation} analyses circuits for implementing unitary complex conjugation.
		Section\,\ref{sec:trans} analyses circuits for unitary transposition.
		Section\,\ref{sec:inverse}  analyses circuits for unitary inversion and Sec.\,\ref{sec:conclusions} concludes and discusses the main results.



\section{Review on higher order quantum operations and supermaps} \label{sec:intro}

	In this section we establish our notations and review how to represent and analyse quantum circuits and transformations between quantum operations in terms of supermaps. We refer to transformations between quantum states as lower order operations (\ie, quantum channels and quantum instruments) and transformations between quantum operations (\eg, channels, instruments, quantum circuits	) as higher order operations, which will be named as superchannels and superinstruements.

	\subsection{The Pills-Choi-Jamio{\l}kowski Isomorphism}
	 	We start by reviewing the Choi isomorphism \cite{pills67,jamiolkowski72,choi75} (also known as Pills-Choi-Jamio{\l}kowski isomorphism), a useful way to represent linear maps and particularly convenient for completely positive ones.
  	   Let $L(\mathcal{H})$ be the set of linear operators mapping a linear (Hilbert) space $\mathcal{H}$ to another space isomorphic to itself. This work only considers finite dimensional quantum systems, hence all linear spaces $\mathcal{H}$  are  isomorphic to $\mathbb{C}^d$, $d$-dimensional  complex linear spaces.	
		Any map\footnote{Symbols with a tilde represent linear maps.} $\map{\Lambda}:L(\mathcal{H}_{\text{in}})\to L(\mathcal{H}_{\text{out}})$ has a one to one representation via its Choi operator defined by,
	\begin{equation}
	\mathfrak{C}(\map{\Lambda}):= \sum_{ij} \map{\Lambda}(\ketbra{i}{j})\otimes \ketbra{i}{j} \in L(\mathcal{H}_{\text{out}}\otimes\mathcal{H}_{\text{in}}),
	\end{equation}
	where $\{\ket{i}\}$ is an orthonormal basis.  An important theorem regarding the Choi representation is that a map $\map{\Lambda}$ is completely positive (CP) if and only if its Choi operator $\mathfrak{C}(\map{\Lambda})$  is positive semidefinite \cite{choi75}.
	When the Choi operator $\mathfrak{C}(\map{\Lambda})$ of some map is given, one can obtain the action of $\map{\Lambda}$ on any operator  $\rho_{\text{in}}\in L(\mathcal{H}_{\text{in}})$ via the relation
		\begin{equation}
	\map{\Lambda}(\rho_{\text{in}})= \tr_{\text{in}} \left( \mathfrak{C}(\map{\Lambda}) \left[ I_{\text{out}} \otimes\;  \rho^T_{\text{in}  }\right] \right),
	\end{equation}
	where $\rho_{\text{in}}^T$ is the transposition of the operator $\rho_{\text{in}}$ in terms of the $\{\ket{i}\}$ basis of $\H_\text{in}$
	
	We now present a useful mathematical identity regarding the Choi isomorphism. Let $U_d$, $A$ and $B$ be  $d$-dimensional unitary operators\footnote{In principle, this identity also holds even when $U_d$, $A$, and $B$ are not unitary but general $d$-dimensional linear operators.}. 
	Any unitary quantum operation  $\map{U_d}(\rho):=U_d\rho U_d^\dagger$ can be represented by its Choi operator as $ \mathfrak{C}(\map{U_d})$ and a straightforward calculation shows that 
	\begin{equation} \label{eq:idenitity_choi}
			\Big[ A\otimes B \Big] \mathfrak{C}(\map{U_d}) \Big[ A^\dagger \otimes B^\dagger \Big]=  \mathfrak{C}(\map{AU_dB^T}).
    \end{equation}
		
	
	\subsection{Supermaps with single use of the input-operations} \label{sec:superchannel}

		In quantum mechanics, physical states are represented by positive operators: $\rho\in L(\mathcal{H})$, $\rho \geq 0$, with unit trace: $\tr{(\rho)}=1$.  	
	In this language, {universal} transformations between quantum states are represented by linear maps, to which we refer as just maps, $\map{\Lambda}: L(\mathcal{H}_{\text{in}})\to L(\mathcal{H}_{\text{out}})$ that are CP \cite{krausbook,chuang}. {Here, by universal we mean that the map $\map{\Lambda}$ is defined for all quantum states $\rho  \in  L(\mathcal{H}_{\text{in}})$ and the physical transformation can be applied to any of these states.} Quantum channels are deterministic quantum operations and are represented by CP maps that preserve the trace of all quantum states $\rho \in L(\H_{\text{in}})$. Probabilistic heralded {universal transformations between quantum states} are represented by quantum instruments, a set of CP maps $\{\map{\Lambda_i}\}$ that sum to a trace preserving one, \ie, $\map{\Lambda}:=\sum_i \map{\Lambda_i}$ is trace preserving (TP). Quantum instruments describe measurements in quantum mechanics%
\footnote{Every instrument $\{\map{\Lambda_i}\}$ corresponds to a unique positive operator valued measure (POVM) measurement $\{M_i\}$, $M_i\in L(\H_{\text{in}}), {M_i\geq 0,} \sum_i M_i = I_{\H{\text{in}}}$, such that  $\tr \left(\rho M_i \right)= \tr\left(\map{\Lambda_i}(\rho)\right)$ for every state $\rho\in L(\H_{\text{in}})$. The POVM $\{M_i\}$ can be written explicitly as $M_i=\map{\Lambda_i^\dagger}(I_{\H_{\text{out}}})$ where $\map{\Lambda_i^\dagger}$ is the adjoint map of $\map{\Lambda_i}$.}.%
		When the set of instruments $\{\map{\Lambda_i}\}$ is performed, the outcome $i$ is obtained with probability $\tr\left(\map{\Lambda_i}(\rho)\right)$, and the state $\rho$ is transformed to $\frac{\map{\Lambda_i}(\rho)}{\tr\left(\map{\Lambda_i}(\rho)\right)}$.

		An important realisation theorem of quantum channels is given by the Stinespring dilation \cite{stinespring55} which states that every quantum channel $\map{\Lambda}$ can be realised by first applying an isometric operation, \ie, a unitary with auxiliary systems and then discarding a part of the system. More precisely, every CPTP map $\map{\Lambda}: L(\H_{\text{\text{in}}})\to L(\mathcal{H}_{\text{out}})$ can be written as
		\begin{equation}
			\map{\Lambda}(\rho)=\tr_{\mathcal{A}'}  \left( U \left[\rho\otimes \sigma_{\mathcal{A}}\right]  U^\dagger\right)
		\end{equation} 
		 where $\sigma_{A}\in L(\mathcal{A})$ is some (constant) auxiliary state, $U :\mathcal{H}_{\text{\text{in}}}\otimes \mathcal{A}\to \mathcal{H}_{\text{\text{out}}}\otimes \mathcal{A}'$ is a unitary acting on the main and auxiliary system,
		 	and $\tr_{\mathcal{A}'}$ is a partial trace on some subsystem $\mathcal{A}'$ such that $\H_{\text{out}}\otimes \mathcal{A}'$ is isomorphic to $\H_{\text{\text{in}}} \otimes \mathcal{A}$.
		 	 
		 Quantum instruments also have an important realisation theorem that follows from Naimark's dillation \cite{naimark40,krausbook}. Every quantum instrument can be realised by a quantum channel followed by a projective measurement, \ie, a measurement which all its POVM elements are projectors, on some auxiliary system. More precisely, if $\{ \map{\Lambda_i}: L(\mathcal{H}_{\text{\text{in}}})\to  L(\mathcal{H}_{\text{out}}) \}$ represents an instrument, there exist a quantum channel $\map{C}:L(\mathcal{H}_{\text{\text{in}}})\to  L(\mathcal{H}_{\text{out}}\otimes \mathcal{A}) $ and a projective measurement given by $\{\Pi_i\}$ where $\Pi_{i} \in L(\mathcal{A})$ which satisfies:
		 \begin{equation}
		 	\map{\Lambda_i}(\rho)=\tr_{\mathcal{A}'} \left({\map{C}}(\rho) \left[ I_{\H_{\text{out}}}\otimes \Pi_i \right] \right).
		 \end{equation}

		 We now define {universal} transformations between quantum operations in an analogous way in terms of linear supermaps \cite{chiribella08}. Linear supermaps, to which we also refer as just supermaps, are linear transformations between maps. A supermap%
		 \footnote{Symbols with a double tilde represent linear supermaps.}, 
		 \begin{equation}
		 	\smap{\phantom{.}S\phantom{.}}: [L(H_2)\to L(H_3)]\to [L(H_1)\to L(H_4)] 
		 \end{equation}
	    	  represents transformations between input-maps $\map{\Lambda_{\text{\text{in}}}}: L(\mathcal{H}_{2})\to L(\mathcal{H}_{3})$
		 	 to output ones ${\map{\Lambda_{\text{out}}}: L(\mathcal{H}_{1})\to L(\mathcal{H}_{4})}$.
	 For instance, let $\map{U_d}(\rho):=U_d \rho U_d^\dagger$ be the map associated to the $d$-dimension unitary operation $U_d$, the supermap that transforms a unitary operation into its inverse is given by	$\smap{S}(\map{U_d})=\map{U^{-1}_d}$.
	 
	 We say that a supermap $\smap{S}$ is TP preserving (TPP) if it transforms TP maps into TP maps. Similarly, a supermap is CP preserving (CPP) when it transforms CP maps into CP maps, and completely CP preserving (CCPP) if the every trivial extension $\smap{S}\otimes \smap{I}$,  of $\smap{S}$ is CPP, where $\smap{I}(\map{\Lambda}) = \map{\Lambda},\; \forall \map{\Lambda}$. 
	 	 A \textit{superchannel} $\smap{C}$ is a supermap which respects two basic constraints: 1) it transforms valid quantum channels into valid quantum channels (hence, CPP and TPP); 2) when  performed to a part of a quantum channel the global channel remains valid (hence, CCPP).

		Any one single use superchannel $\smap{C}$ has a deterministic realisation in quantum theory and similarly to the Stinespring dilation theorem, it can be shown that every superchannel admits a decomposition in terms of \textit{encoder} and \textit{decoder} of the form \cite{chiribella08}, 
		\begin{equation}
		\smap{C}( \map{\Lambda})=\map{D} \circ \left[ \map{\Lambda} \otimes   \map{I_A} \right] \circ \map{E}
		\end{equation} where $\map{E}:L(\H_1)\to L(\H_2)\otimes L(\H_A) $ is an isometry which maps an input-state  $\rho_\text{in} \in L(\mathcal{H}_1)$ to the space where the map $\map{\Lambda}$ acts and an auxiliary one $L(\mathcal{H}_A)$, $\map{I_\mathcal{A}}$ is the identity map on the auxiliary system (\ie $\map{I_A}(\sigma_A) = \sigma_A, \; \forall \; \sigma_A\in L(\mathcal{H}_A$), $\map{D}:L(\H_3\otimes \H_A)\to L(\H_4)$ is a unitary operation followed by a partial trace on a part of the system (see  Fig. \ref{fig:encoder_decoder}).  
		
		\begin{figure}  
\begin{center}
\includegraphics[scale=.18]{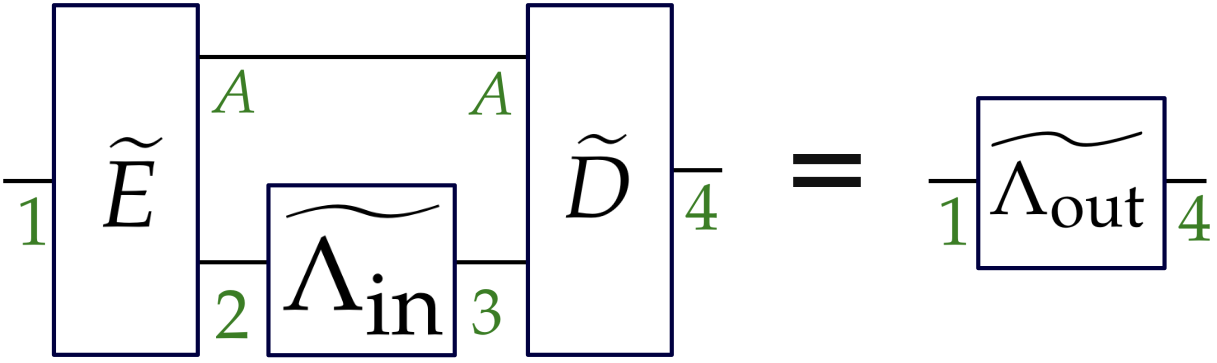} 
\end{center}
\caption{Every superchannel $\smap{C}: [L(\H_2)\to L(\H_3)]\to [L(\H_1)\to L(\H_4)] $ transforming input-channels $ \map{\Lambda_{{\text{\text{in}}}}}: L(\mathcal{H}_{2})\to L(\mathcal{H}_{3})$ into output channels $ \map{\Lambda_{{\text{out}}}}: L(\mathcal{H}_{1})\to L(\mathcal{H}_{4})$ can be decomposed as $	\smap{C}( \map{\Lambda_{{\text{\text{in}}}}})= \map{D} \circ \left[ \map{\Lambda_{{\text{\text{in}}}}} \otimes   \map{I_A} \right]  \circ \map{E}  $ where the encoder operation $\map{E}:L(\H_1)\to L(\H_2)\otimes L(\H_A) $ is an isometry operation,  $L(\mathcal{H}_A)$ is a space for some possible auxiliary system, and the decoder $\map{D}:L(\H_3\otimes \H_A)\to L(\H_4)$ is a unitary operation followed by a partial trace on a part of the system.}\label{fig:encoder_decoder}
\end{figure}
	
	 	The Choi representation allows us to describe any supermap $\smap{S}:[L(\H_2)\to L(\H_3)]\to [L(\H_1)\to L(\H_4)]$ as a map $\map{S}:L(\mathcal{H}_3\otimes\mathcal{H}_2) \to L(\mathcal{H}_4\otimes\mathcal{H}_1) $ acting on Choi operators.
		And by exploiting the Choi representation  again, we can represent any supermap $\smap{S}$ by a linear operator $S:=\mathfrak{C}(\map{S}) \in L(\mathcal{H}_4 \otimes \mathcal{H}_1 \otimes  \mathcal{H}_3 \otimes  \mathcal{H}_2)$, which is useful to characterise the set of supermaps with quantum realisations. In Ref.\,\cite{chiribella07,chiribella08} the authors show that a $\smap{C}$ is a superchannel if and only if its Choi representation $C$ respects 
		\begin{equation} \label{eq:superchannel}
		\begin{split}
		C &\geq0 ;\\
			\tr_{4} C &= \tr_{43} C  \otimes  \frac{I_3}{d_3}  ; \\
			\tr_{234} C &= \tr_{1234} C  \otimes  \frac{I_1}{d_1}  ; \\
			\tr(C) &= d_1 d_3, 
			\end{split}
	\end{equation}
	where $d_i$ is the dimension of the linear space $\H_i$. {We remark that although we introduce the general formalism where the dimensions $d_i$ may depend on $i$, we focus our results to the case where $d_i=d$ is independent of $i$.}

	 Supermaps with probabilistic a heralded quantum realisation are given by \textit{superinstruments} and play a similar role of instruments in {higher order quantum operations, that is, it formalises probabilistic transformations on quantum operations.} Superinstruments are a set of CCPP supermaps $\{\smap{C_i}\}$ that sums to a superchannel. The probability of obtaining the outcome $i$ when the superinstrument $\{ \smap{C_i} \}$ acts on the input-map $\map{\Lambda}$ and input-state $\rho$ is $\tr \left( \left[ \smap{C_i}\left(\map{\Lambda}\right)\right] (\rho) \right)$ and the state $\frac{ \left[ \smap{C_i}\left(\map{\Lambda}\right)\right]  (\rho) }{\tr \left( \left[ \smap{C_i}\left(\map{\Lambda}\right) \right] \; (\rho) \right)}$ is obtained. It follows from Ref.\,\cite{bisio16} that any superinstrument can be realised by a superchannel followed by a projective measurement, or equivalently,
	 \begin{equation} \label{eq:realise_superis}
		\smap{C_i}( \map{\Lambda})=  \map{D_i} \circ \left[ \map{\Lambda} \otimes   \map{I_A} \right] \circ \map{E},
	\end{equation} 
here $\map{E}:L(\H_1)\to L(\H_2)\otimes L(\H_A) $ is an isometry which maps an input-state  $\rho \in L(\mathcal{H}_1)$ to the space where the map $\map{\Lambda}$ acts and an auxiliary one $L(\mathcal{H}_A)$, $\map{I_A}$ is the identity map on the auxiliary system (\ie, $\map{I_A}(\sigma_A) = \sigma_A, \; \forall \; \sigma_A\in L(\mathcal{H}_A)$), and the maps $\map{D_i}:L(\H_3\otimes \H_A)\to L(\H_4)$ form an instrument corresponding to a projective measurement.


  \subsection{Supermaps involving $k$ input-operations} \label{sec:supermaps}
  
  		In the previous section we have introduced supermaps corresponding to protocols involving a single use of an input-operation. We now consider protocols transforming $k$, potentially different, operations into another.
  	  	Let $\smap{C}$ be a superchannel which transforms $k$ input-channels\footnote{We 
  	  	remark that here the subindex $j$ stands for a label for the channel  	  	${\map{\Lambda_j}: L(\mathcal{I}_j) \to L(\mathcal{O}_j)}$, 	  	 not for some instrument element of an instrument $\{\Lambda_j\}$.} 
  	  	  $\map{\Lambda_j}: L(\mathcal{I}_j) \to L(\mathcal{O}_j)$  with $j\in\{1,\ldots,k\}$ into an output one $\map{\Lambda_0}:L(\mathcal{I}_{0})\to L(\O_0)$.  We also define the total input-state space as $\mathcal{I}:=\bigotimes_{j=1}^{k}  \mathcal{I}_j$  and the total output space state $\mathcal{O}:=\bigotimes_{j=1}^{k} \mathcal{O}_j$, hence $\smap{C}: \left[L(\mathcal{I})\to L(\mathcal{O})\right] \to \left[L(\mathcal{I}_0)\to L(\mathcal{O}_0)\right] $.
 

  	\begin{figure}  
	\begin{center}
	\includegraphics[scale=.2]{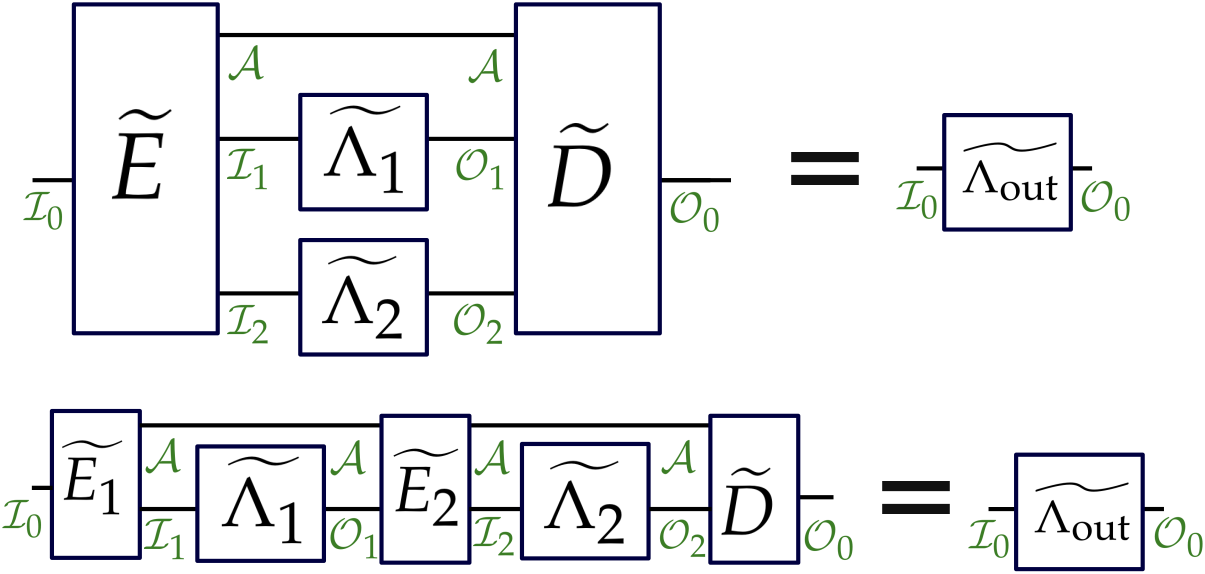} 
	\end{center}
			\caption{Illustration of a parallel (upper circuit) and an adaptive (lower circuit) protocol that transform a pair of quantum operations $\map{\Lambda_1}$ and $\map{\Lambda_2}$ into an output one $\map{\Lambda_{\text{out}}}$.} \label{fig:adapt_non_adapt}
	\end{figure}
  	

  	  	Similarly to the single input-channel case, superchannels transforming $k$ quantum operations are supermaps which: 1) transform $k$ valid quantum channels into a valid quantum channel; 2) when  performed on a part of a quantum channel, the global channel remains valid. Differently from the $k=1$ case, not all superchannels have a deterministic quantum realisation in terms of encoders and decoders in the standard quantum circuit formalism \cite{chiribella09}. This impossibility occurs because the definition of quantum realisation does not require explicitly that the $k$ channels should be used in a definite causal order and it allows protocols which use the input-channels with an indefinite causal order \cite{oreshkov11}.
  	  	
  	  	Protocols that can be implemented in the standard causally ordered circuit formalism are referred to as quantum networks/quantum combs \cite{chiribella07} or channels with memory \cite{kretschmann05}. We divide these ordered circuits in two classes: a) parallel ones where $k$ channels can be used simultaneously; b) adaptive ones where the $k$ channels are explored in a causal sequential circuit (see Fig.\,\ref{fig:adapt_non_adapt}).

  		Parallel protocols transforming $k$ channels are very similar to single-channel superchannels presented in the last subsection. Define $\map{\Lambda}:=\bigotimes_{j=1}^{k}  \map{\Lambda_j}$, a superchannel $\smap{C}$ represents a parallel protocol if it can be written as $\smap{C}(\map{\Lambda})=\map{D}\circ \left[ \map{\Lambda}\otimes \map{I_\mathcal{A}} \right] \circ \map{E}$ for some channels $\map{E}:L(\mathcal{I}_0)\to L(\mathcal{I}\otimes \mathcal{A})$ and $\map{D}:L(\O\otimes \mathcal{A})\to \O_0$. It follows from the characterisation of Eq.\,\eqref{eq:superchannel} that a Choi operator 	$C\in L(\mathcal{I}_0 \otimes \bigotimes_{j=1}^{k} \mathcal{I}_j \otimes \bigotimes_{j=1}^{k} \O_j  \otimes \O_0  )$ represents a parallel protocol if and only if
   	  	
  	\begin{equation}  \label{eq:non_adaptive}
  	 \begin{split}
	C &\geq0 ;\\
	\tr_{\O_0} C &= \tr_{\O\O_0} C \otimes \frac{I_{\mathcal{O}}}{d_{\mathcal{O}}}  ;  \\
	 \tr_{\mathcal{I}\O\O_0} C &=  \tr_{\mathcal{I}_0\mathcal{I}\O\O_0} C  \otimes \frac{I_{\mathcal{I}_0}} {d_{\mathcal{I}_0}} ;  \\
	\tr(C) &= d_{\mathcal{I}_0} d_{\mathcal{O}}.  
	\end{split}  
	\end{equation}

		Adaptive circuits can exploit a causal order relation between the channels $\map{\Lambda_j}$ to implement protocols that cannot be done in a parallel way. A simple example is the supermap that concatenates the channels 
	$\map{\Lambda_1}$ and  $\map{\Lambda_2}$ to obtain  $ \map{\Lambda_2} \circ  \map{\Lambda_1}  $. This supermap has a trivial implementation in an adaptive circuit (just concatenates the channels) but cannot be implemented in a \textit{deterministic} parallel scheme.

		A superchannel ${\smap{C}: \left[ L(\mathcal{I})\to L(\mathcal{O})\right] \to \left[ L(\mathcal{I}_0)\to L(\mathcal{O}_{0})\right]}$ corresponds to an adaptive  circuit if it can be written as\footnote{Note that since we do not restrict the dimension of the auxiliary system $\mathcal{A}$, all parallel protocols can be realised by an adaptive  circuit.}
		\begin{equation}
		\smap{C}(\map{\Lambda})=\map{D}\circ \left[ \map{\Lambda_k} \otimes \map{I_\mathcal{A}} \right] \circ  \map{E_k} \circ \ldots \circ \left[ \map{\Lambda_1} \otimes \map{I_\mathcal{A}} \right]\circ \map{E_1}
		\end{equation}
		for some channels $\map{E_1}:L(\mathcal{I}_0)\to L(\mathcal{I}_1\otimes \mathcal{A})$, $\map{E_i}:L(\O_{i-1}\otimes \mathcal{A})\to (\mathcal{I}_{i}\otimes \mathcal{A})$ with $i\in\{2,\ldots,k\}$, and $\map{D}:L(\O_k\otimes \mathcal{A})\to L(\O_0)$.
		A Choi operator	$C\in L(\mathcal{I}_0 \otimes \bigotimes_{j=1}^{k} \mathcal{I}_j \otimes \bigotimes_{j=1}^{k} \O_j  \otimes \O_0  )$ represents an adaptive  superchannel if and only if \cite{chiribella07,kretschmann05}
		\begin{equation}  \begin{split} \label{eq:adaptive}
	C &\geq0 ;\\
	\tr_{\O_0} C &= \tr_{\mathcal{O}_k\O_0} C \otimes \frac{I_{\mathcal{O}_k}}{d_{\mathcal{O}_k}}  ;  \\
	\tr_{\mathcal{I}_{i}} C^{(i)} &= 	\tr_{\mathcal{I}_{i} \mathcal{O}_{i} } C^{(i)} \otimes \frac{ I_{\mathcal{O}_{i}}}{d_{\mathcal{O}_{i}}},  \quad \forall i \in \{k,\ldots,2\};   \\
	\tr_{\mathcal{I}_1} C^{(1)} &= 	\tr_{\mathcal{I}_0 \mathcal{I}_1} C^{(1)} \otimes \frac{ I_{\mathcal{I}_0}}{d_{\mathcal{I}_0}}	\\
	\tr(C) &=  d_{\mathcal{I}_0} d_{\mathcal{O}},   
	\end{split}  \end{equation}  
	where $C^{(i)}:=\tr_{\mathcal{I}_{i+1} \mathcal{O}_{i+1}} C^{(i+1)}  $ for $i\in \{1,\ldots,k-1\}$ and $C^{(k)}:=\tr_{\O_k \mathcal{O}_0} C$.

			We now consider the most general protocols that transform $k$ quantum channels into a single one. As mentioned before, these superchannels may have an indefinite causal order between the use of these $k$ channels, hence they may not have an implementation in terms of encoders and decoders in the standard quantum circuit formalism. Even without necessarily having a realisation by ordered circuits, it is possible to have a simple characterisation of these general superchannels. Before presenting the necessary and sufficient condition for a general (possibly with an indefinite causal order) superchannels, it is convenient to introduce the trace and replace notation introduced in Ref.\,\cite{araujo15}. Let $A\in L(\H_1\otimes \H_2)$ be a general linear operator, we define $_{\mathcal{H}_2} A:=\tr_{\H_2} A\otimes\frac{I_{\H_2}}{d_{\H_2}}$. A Choi operator $C\in L(\mathcal{I}_0 \otimes \mathcal{I}_1\otimes \mathcal{I}_2 \otimes \O_1 \otimes \O_2  \otimes \O_0  )$ represents a general superchannel transforming $k=2$ channels into a single one if and only if it respects \cite{araujo16}:%
		%
	%
		\begin{equation}  \begin{split} \label{eq:non-causal}
	C &\geq0 ;\\
		{}_{\mathcal{I}_1\mathcal{O}_1  \O_0} C &= {}_{\mathcal{I}_1\mathcal{O}_1\mathcal{O}_2   \O_0}  C   ;  \\
		{}_{\mathcal{I}_2 \mathcal{O}_2 \O_0} C &={}_{\mathcal{O}_1\mathcal{I}_2 \mathcal{O}_2 \O_0}  C   ;  \\
	{}_{\O_0} C+ {}_{\mathcal{O}_1\mathcal{O}_2  \O_0}  C &= {}_{\mathcal{O}_1  \O_0} C  + {}_{\mathcal{O}_2  \O_0} C  ;  \\
	{}_{\I \O \O_0} C &= {}_{\I_0\I \O \O_0} C  ; \\   
	\tr(C) &= d_{\mathcal{I}_0}d_{\O_1}d_{\O_2}  .  
	\end{split}  \end{equation}  
		We remark that the bipartite process matrices presented in Ref.\,\cite{oreshkov11,araujo15} correspond to a particular case of general superchannels with two input-channels presented above. This correspondence is made by setting the dimension of the linear spaces  $\mathcal{I}_0$ and $\mathcal{O}_0$ as one. This occurs because Ref.\,\cite{oreshkov11,araujo15} focus on superchannels that transform pairs of instruments into probabilities, not into quantum operations. 
	{Also, the general superchannels presented in Eq.\,\eqref{eq:non-causal} are equivalent to the general process matrices presented in Ref.\,\cite{araujo16} which uses the terminology common past and common future to denote the spaces $\mathcal{I}_0$ and $\mathcal{O}_0$, respectively.}
		
		{It is also possible to characterise general superchannels transforming $k$ channels on terms of their Choi operators. For that, one can exploit the methods used in Ref.\,\cite{araujo16} and \cite{araujo15} to characterise process matrices (see also Ref.\,\cite{chiribella16}) . Using such methods, we have characterised general superchannels which transforms $k=3$ input-channels into a single output one. A Choi operator $C\in L(\mathcal{I}_0 \otimes \mathcal{I}_1\otimes \mathcal{I}_2 \otimes \mathcal{I}_3 \otimes \O_1 \otimes \O_2 \otimes \O_3   \otimes \O_0  )$ represents a general superchannel that transforms $k=3$ channels into another one if and only if it respects	}	
			\begin{equation}  \begin{split} \label{eq:non-causal3}
	C &\geq0 ;\\
	{}_{ \mathcal{I}_1\O_1\mathcal{I}_2\O_2\O_0} C &={}_{ \mathcal{I}_1\O_1\mathcal{I}_2\O_2 \O_3 \O_0} C   ;  \\
	{}_{\mathcal{I}_2\O_2\mathcal{I}_3 \mathcal{O}_3 \O_0} C &={}_{\mathcal{O}_1 \mathcal{I}_2\O_2\mathcal{I}_3 \mathcal{O}_3 \O_0}  C   ;  \\
	{}_{ \mathcal{I}_1\O_1\mathcal{I}_3\O_3\O_0} C &={}_{\mathcal{I}_1\O_1\O_2\mathcal{I}_3\O_3\O_0} C   ;  \\
	{}_{\mathcal{I}_1\O_1\O_0}C + {}_{\mathcal{I}_1\O_1\O_2\O_3\O_0}C &= {}_{\mathcal{I}_1\O_1\O_2\O_0}C + {}_{\mathcal{I}_1 \O_1 \O_3\O_0}C ;  \\
    {}_{\mathcal{I}_2\O_2\O_0}C + {}_{\O_1\mathcal{I}_2\O_2\O_3\O_0}C &= {}_{\O_1\mathcal{I}_2\O_2\O_0}C + {}_{\mathcal{I}_2\O_2\O_3\O_0}C ;  \\
	{}_{\mathcal{I}_3\O_3\O_0}C + {}_{\O_1\O_2\mathcal{I}_3\O_3\O_0}C &= {}_{\O_1\mathcal{I}_3\O_3\O_0}C + {}_{\O_2\mathcal{I}_3\O_3\O_0}C ;  \\
		{}_{\O_0}C + {}_{\O_1\O_2\O_3\O_0}C =&  {}_{\O_1\O_0}C + {}_{\O_2\O_0}C + {}_{\O_3\O_0}C +    \\
		+{}_{\O_1\O_2\O_0}C  +& {}_{\O_1\O_3\O_0}C +  {}_{\O_2\O_3\O_0}C ;   \\
{}_{\I \O \O_0} C &= {}_{\I_0\I \O \O_0} C \; ; \\  
	\tr(C) &= d_{\mathcal{I}_0}d_{\O_1}d_{\O_2} d_{\O_3} .  
	\end{split}  \end{equation}

		Similarly to the single use case, probabilistic heralded protocols are also represented by superinstruments. Superinstruments also admit a simple representation in terms of their induced Choi operators. A set of parallel/adaptive /general superinstruments transforming $k$ channels into another is given by a set of positive semidefinite operators $C_i\geq 0$ where $C:= \sum_i C_i$ is a valid parallel/adaptive /general superchannel. The probability of obtaining the outcome $i$ when performing the superinstrument $\left\{\smap{C_i} \right\}$ on $k$ input-channels represented by    
		$\map{\Lambda}:=\bigotimes_{j=1}^k \map{\Lambda_j}$ and the input-state $\rho$ is given by   %
		$\tr \left(  \left[ \smap{C_i}\left(\map{\Lambda}\right) \right] (\rho) \right)$.



\section{Optimal universal quantum circuits via SDP} \label{sec:SDP}

		In this section we construct a systematic method to design probabilistic heralded quantum circuits for transforming multiple uses of the same unitary operations. {Let  $U_d:L(\mathbb{C}^d)$ be a $d$-dimensional unitary operator and  $\map{U_d}$ be a linear map representing the operation associated to $U_d$, that is, when the operation $\map{U_d}$ is applied into a quantum state $\rho \in L(\mathbb{C}^d)$ the output is given by $\map{U_d}(\rho)=U_d\rho U_d^\dagger$. We consider linear supermaps given by $f:\map{U_d}\mapsto f(\map{U_d})$  mapping unitary operations into unitary operations.} Our goal is to transform $k$ uses of an arbitrary $\map{U_d}$ into ${f}(\map{U_d})$ with the highest heralded constant probability $p$. 
		
		From the results of the previous section, this transformation can be implemented via quantum circuits when there exists a superinstrument element \ie, a CCPP linear supermap, $\smap{S}$ such that $\smap{S}(\map{U_d^{\otimes k}}) = p \,	 f(\map{U_d})$ for every unitary operation $\map{U_d}$ (see Sec. \ref{sec:superchannel}).  {We stress that, even thought we have presented an explict characterisation of superinstruments in Sec.\,\ref{sec:superchannel}, finding the optimal success probability for this transformation and its associated quantum circuit is, in general, a nontrival task.} First, note that action of the supermap $f$ is only described for unitary channels\footnote{Since we have imposed that $f$ is linear, $f$ is also implicitly defined for linear combination of unitary operations.} but the action of a superinstrument element $\smap{S}$ must be defined for any CP linear map. The supermap  $\smap{S}$ can then be any CCPP linear supermap that extends the action of $f$ from unitary operations to general CP maps (see Ref.\,\cite{chefles03,heinosaari12} for a related lower-order version problem which consists of finding CP extentions of linear maps defined on subspaces).  Second, since $k$ uses of the input-operation are available, it may be the case that even if $f$ does not have a linear CCPP extention for some number of uses $k_0$ but it has for $k>k_0$ (see Ref.\,\cite{buzek99,dong18} for a lower-order analogue of this problem where multiple copies of the input-\textit{state} can be used to implement a linear positive non-CP \textit{map}).
		
		 Before presenting our general approach we illustrate the subtleties of this extention problem by discussing the universal channel complex conjugation studied in Ref.\,\cite{miyazaki17}.
		 Let $\map{\Lambda} : \H_{\text{in}}\to  \H_\text{out}$ be a quantum channel with the Kraus decomposition given by $\map{\Lambda}(\rho)=\sum_i K_i \rho K_i^\dagger$, we define the complex conjugate of 
		$\map{\Lambda}$ as the map which respects $\map{\Lambda^*}(\rho)=\sum_i K_i^* \rho K_i^{*^\dagger}$ for every $\rho$ where the complex conjugation of $K_i$ is made in a fixed orthonormal basis \eg, the computational basis. One can show that, for any linear spaces $\H_\text{in}$ and $\H_\text{out}$ with dimension greater than or equal two,  CCPP supermaps respecting $\map{\Lambda}^{\otimes k} \mapsto p \map{\Lambda^*}$ for all channels $\map{\Lambda}$ necessarily have $p=0$ for any number of uses $k\in\mathbb{N}$ \cite{miyazaki_thesis}. Hence it is not possible to design a universal quantum circuit for probabilistic channel adjoint. However, if one relaxes the requirements of general channels and seek for a quantum circuit that transforms only unitary operations into their adjoints, universal complex conjugation can be implemented deterministically in a parallel circuit with makes $k=d-1$ uses of the input-channel presented in Ref.\,\cite{miyazaki17}. In Sec.\ref{sec:conjugation} we prove that $k=d-1$ uses are not only sufficient but also necessary. We then see that the notion of CCPP extension and the number of uses play a crucial role in finding superinstruments that implement some desired transformation given by $f$.

	    We now present our SDP approach.
		 Let $\left\{ \smap{S},\smap{F} \right\}$ be a superinstrument where the outcome of the element $\smap{S}$ indicates success and the outcome of $\smap{F}$ indicates failure.
	The problem of maximising the success probability of transforming $k$ uses of an arbitrary $d$-dimensional unitary input-operation $\map{U_d}$ into $f(\map{U_d})$ can be phrased as:
		\begin{equation}  \begin{split}
		&\text{max } p  \\ 
		\text{s.t. }  \quad & \smap{S} \left(\map{U_d}^{\otimes k} \right) = p f(\map{U_d}), \quad \; \forall U_d;   \\ 
		&\left\{ \smap{S},\smap{F} \right\} \text{ is a valid superinstrument},
	\end{split}  \end{equation}  
	where the valid superinstrument representing a  parallel, adaptive, or general protocols.
	Using the characterisation presented in Sec.\,\ref{sec:intro}, we can rewrite the above maximisation problem only in terms of linear and positive semidefinite constraints as:
\begin{small}
		\begin{equation}  \begin{split} \label{eq:SDP}
						 &\text{max } p  \\ 
		\text{s.t. }  \quad & \tr_{\mathcal{IO}}  \left( S \left[ I_{\mathcal{I}_0} \otimes \C \left(\map{U_d^{\otimes k}}\right) ^T_{\mathcal{I}\O} \otimes I_{\O_0} \right]\right) = p \C(f(\map{U_d})) \; \forall U_d;   \\ 	
		&S,F\in L(\mathcal{I}_0 \otimes \mathcal{I} \otimes \mathcal{O} \otimes \mathcal{O}_0 ),\quad  S \geq0, \; F\geq 0   ;\\ 
		&S+F  \text{ is a valid superchannel}   .
	\end{split}  \end{equation}  
\end{small}
		Note that the maximisation problem presented in Eq.\,\eqref{eq:SDP} must hold for all unitary operators $U_d$ and has infinitely many constraints.  This issue can be bypassed by noting that due to linearity, it is enough to check these constraints only for a set that spans the set spanned by Choi operators of unitary operations. That is, if we can write 
		$\C(\map{U_d^{\otimes k}})= \sum_i \alpha_i\, \C(\map{U_{d,i}^{\otimes k}})$ and  it is true that%
				\begin{equation}  \begin{split} 
	 \tr_{\mathcal{IO}}  \left( S \left[ I_{\mathcal{I}_0} \otimes \C \left(\map{U_{d,i}^{\otimes k}}\right) ^T_{\mathcal{I}\O} \otimes I_{\O_0} \right]\right) = p \C(f(\map{U_{d,i}})) \; \forall i;
	\end{split}  \end{equation}  
		we have that
				\begin{equation}  \begin{split} \label{eq:SDP_proof}
\tr_{\mathcal{IO}}  &\left( S \left[ I_{\mathcal{I}_0} \otimes \C \left(\map{U_d^{\otimes k}}\right) ^T_{\mathcal{I}\O} \otimes I_{\O_0} \right]\right)  \\ = &
\sum_i  \tr_{\mathcal{IO}}  \left( S \left[ I_{\mathcal{I}_0} \otimes \alpha_i  \C \left( \map{U_{d,i}^{\otimes k}}\right) ^T_{\mathcal{I}\O} \otimes I_{\O_0} \right]\right)  \\=
 & p\sum_i \alpha_i   \C(f(\map{U_{d,i}})) \\=
 &  p\C(f(\map{U_{d}})).
	\end{split}  \end{equation}  
		Also,  one can always find a finite set, in particular, a basis, of unitary operations $\{\map{U_{d,i}}\}$ that spans the set spanned by Choi operators of $d$-dimensional unitary operations, \ie:
		\begin{equation}  \begin{split}
		\text{span}\Big(\C(\map{U_d^{\otimes k}}) & \;| \; U_d \text{ is unitary} \Big) =  \\
	&	\text{span} \left(\C(\map{U_{d,i}^{\otimes k}}) \; | \;U_{d,i} \in \{ U_{d,i} \}\right). \\	
		\end{split}  \end{equation}  

		Explicitly obtaining a basis for the subspace $\text{span}\left(\C(\map{U_d^{\otimes k}}) \;| \;U_d \text{ is unitary}  \right)$ is, in general, not straightforward. For numerical purposes, this problem can be tackled by sampling a large number of unitaries $U_d$ uniformly randomly (according to the Haar measure). If the dimension of this subspace is $D$, $D$ unitaries sampled uniformly will be linearly independent with unit probability. Since checking linear independence can be done in an efficient way, we can construct a basis for this set by sampling unitaries randomly until we cannot find more linearly independent ones. 
		
		Also note that the dimension $D$ of the subspace $\text{span}\left(\C(\map{U_d^{\otimes k}}) \;| \;U_d \text{ is unitary}  \right)$ may grow very fast with $k$ and $d$, this will increase the number of constraints in the SDP we have presented. Since having a large number of constraints may make the SDP intractable for practical purposes (it may take a very long time to run the code or to consume a very large amount of Random-access memory (RAM)), it is worth noticing that if one runs the SDP \eqref{eq:SDP} with a set of operators $\{U_{d,i}\}$ that do not form a basis for $\text{span}\left(\C(\map{U_d^{\otimes k}}) \;| \;U_d \text{ is unitary}  \right)$, the solution $p$ of the SDP is not the maximal success probability but an upper bound on the maximal success probability (it is the same SDP with fewer constraints). We also point out that since the methods to solve an SDP also provide the instrument element $S$ that attains the maximal success probability $p$, even if the set $\{U_{d,i}\}$ does not form a basis for $\text{span}\left(\C(\map{U_d^{\otimes k}}) \;| \;U_d \text{ is unitary}  \right)$, it may still be the case that the solution obtained {is also the global optimal value}\footnote{We thank Alastair Abbott for pointing this fact to us.}. In order to check this hypothesis we can extract the superinstruement element $S$ of the SDP in which the operators $\{U_{d,i}\}$ that do not form a basis for $\text{span}\left(\C(\map{U_d^{\otimes k}}) \;| \;U_d \text{ is unitary}  \right)$. Then, we generate a basis  $\{U'_{d,j}\}$ for $\text{span}\left(\C(\map{U_d^{\otimes k}}) \;| \;U_d \text{ is unitary}  \right)$ and verify that 
		\begin{equation}
		\tr_{\mathcal{IO}}  \left( S \left[ I_{\mathcal{I}_0} \otimes \C \left(\map{U_{d,j}^{'\otimes k}}\right) ^T_{\mathcal{I}\O} \otimes I_{\O_0} \right]\right) = p \C(f(\map{U'_{d,j}}))  
		\end{equation}
		for every\footnote{When $d=3$, $k=2$, we have applied this technique to tackle the unitary transposition and inversion problem. In this case, we have run our numerical SDPs only for a subset of the space of unitary channels generated by $\text{span}\Big(\C(\map{U_3^{\otimes 2}})  \;| \; U_3 \text{ is unitary} \Big) $. Numerically, we can see that the linear space spanned by $\Big(\C(\map{U_3^{\otimes 2}})  \;| \; U_3 \text{ is unitary} \Big) $ has $994$ linearly independent  unitary channels but we have only considered a random subset containing $200$ linear independent elements of the form $\C(\map{U_3^{\otimes 2}})$ in our calculations. After obtaining an upper bound to the problem, we have verified that the superinstrument element $S$ transforms the full basis with  $994$ linearly independent  unitary channels into their inverses, ensuring that the previous upper bound is tight.}  $j$.

		For the particular cases where the desired operation is unitary inversion, \ie, $f(\map{U_d})=\map{U_d^{-1}}$ or  the desired operation is unitary transposition	, \ie, $f(\map{U_d})=\map{U_d^{T}}$, the optimal success probability $p$ is always attainable by instruments where the Choi operators $S$ and $F$ only have real number components. 
 That is, the operators $S$ and $F$ can be restricted to the field of real numbers with no loss of generality. To prove that, we first note that for every unitary $U_d$ we have 
	 $\C \left(\map{U_d^{\otimes k}}\right)^* = \C \left(\map{{U_d^*}^{\otimes k}}\right) $
	where ${}^*$ is complex conjugation in the computational basis. We present the explicit proof for the unitary transposition case and the unitary inverse follows from the same steps. Assume that there exists a superinstrument with a success Choi operator $S$ such that
	\begin{equation}
\tr_{\mathcal{IO}}  \left( S \left[ I_{\mathcal{I}_0} \otimes \C \left(\map{U_d^{\otimes k}}\right) ^T_{\mathcal{I}\O} \otimes I_{\O_0} \right]\right) = p \C(\map{U_d^T}) 
	\end{equation} 
	holds for all unitaries $U_d$. Direct calculation shows that the instrument defined by $S^*$ attains the same performance of $S$:
	\small
	\begin{align}
	& \tr_{\mathcal{IO}}  \left( S^* \left[ I_{\mathcal{I}_0} \otimes \C \left(\map{U_d^{\otimes k}}\right) ^T_{\mathcal{I}\O} \otimes I_{\O_0} \right]\right) \\
	=&	 {{\tr_{\mathcal{IO}}  \left( S^* \left[ I_{\mathcal{I}_0} \otimes \C \left(\map{U_d^{\otimes k}}\right) ^T_{\mathcal{I}\O} \otimes I_{\O_0} \right]\right)}^*}^* \\
	=&	 {{\tr_{\mathcal{IO}}  \left( {S^*}^* \left[ I_{\mathcal{I}_0}^* \otimes \C {\left(\map{U_d^{\otimes k}}\right) ^T_{\mathcal{I}\O}}^* \otimes I_{\O_0}^* \right]\right)}^*} \\ 
	=&	 {\tr_{\mathcal{IO}}  \left( S \left[ I_{\mathcal{I}_0} \otimes \C {\left(\map{{U_d^*}^{\otimes k}}\right) ^T_{\mathcal{I}\O}} \otimes I_{\O_0} \right]\right)}^* \\ 
	=& p \C {\left(\map{{{U_d^*}^T}}\right)} ^* \\
	=& p \C \left(\map{U_d^T}\right).
	\end{align}
\normalsize
	Since $\{S^*, F^*\}$ represents a valid superinstrument, we can construct the operators $S':=\frac{S+S^*}{2}$ and $F':=\frac{F+F^*}{2}$ which only have real number components.

		We have implemented our code using MATLAB \cite{MATLAB} with the interpreter CVX \cite{cvx} and tested with the solvers MOSEK, SeDuMi, and SDPT3 \cite{MOSEK,SEDUMI,SDPT3}.  In Table \ref{table:SDP2} of Sec.\,\ref{sec:SDP_trans} we apply this method to obtain the maximal success probability to transform $k$ uses of a $d$-dimensional unitary operation, \ie,  $f(\map{U_d})=\map{U_d^T}$ under different constraints. In Sec.\,\ref{sec:SDP_inv}, we reproduce Table 1 of Ref.\,\cite{PRL} which contains results for the maximal probability for unitary inversion \ie,  $f(\map{U_d})=\map{U_d^{-1}}$ .		
	All our code are available at Ref.\,\cite{MTQ_github_unitary} and can be freely used, edited, and distributed under the MIT license \cite{MIT_license} and  make extensive use of the toolbox QETLAB \cite{QETLAB}.

 \section{ Delayed input-state protocols} \label{sec:delayed}

		In this section we define a particular subclass of quantum circuits in which we refer to delayed input-state protocols. This class consists of circuits where the input-state is provided after the input-operation which will be transformed (see Fig.\,\ref{fig:delayed}). The concept of delayed-input-state generalises the class of supermaps considered in the context of unitary learning and unitary store-and-retrieve problems \cite{vidal01,sasaki02,gammelmark09,bisio10,sedlak18}. As we will show next, parallel quantum circuits used for unitary transposition and unitary inversion can be assumed to be in the delayed input-state form without loss of generality and the definition of delayed input-state protocols is useful to prove various theorems presented in this paper.
		
		  Consider a scenario where Alice has $k$ uses of a general unitary operation $\map{U_d}$ until some time $t_1$. In a later time $t_2$, where $\map{U_d}$ cannot be accessed anymore, she would like to implement $f(\map{U_d})$ on some arbitrary quantum state chosen at time $t_2$. This scenario can be seen as a particular case of the general unitary transformation problem where the input-state is only provided after the operation $\map{U_d}$.		
		Let us start with the $k=1$ case where only a single use of the general input-operation $\map{\Lambda_\text{in}}$ is allowed (see Fig.\,\ref{fig:delayed}). 
		%
		 	%
  	%
  	%
	\begin{figure}
	\begin{center}
	\includegraphics[scale=.18]{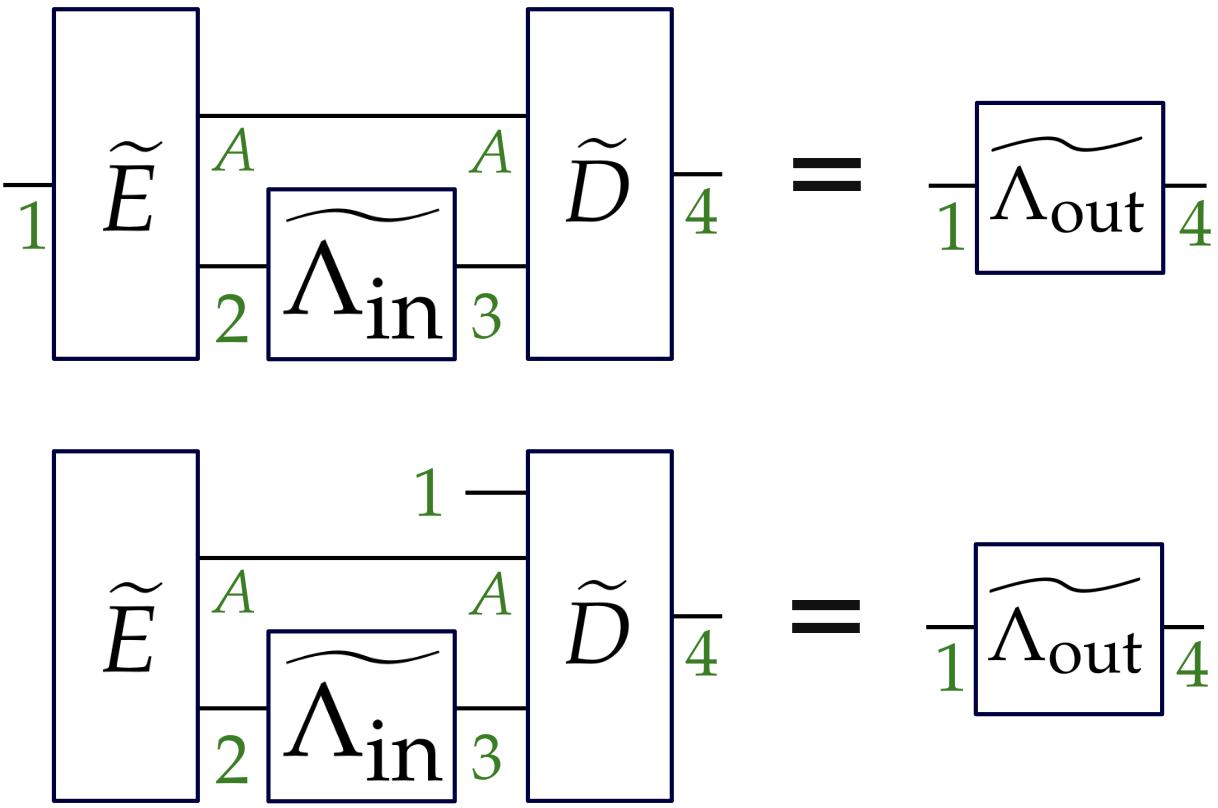}
	\end{center}			\caption{Comparison between a standard quantum circuit (upper circuit) and a delayed input-state protocol (lower circuit) that transforms general operations. In a delayed input-state protocol, the input-state labelled by the space $1$ is not used by the encoder operation $\map{E}$. The encoder only prepares a (potentially entangled) state which partially goes to the input-channel $\map{\Lambda_{\text{\text{in}}}}$, and {then to the decoder channel $\map{D}$,} which can perform a joint operation between the input-state and the auxiliary system. }  \label{fig:delayed}
	\end{figure}
  	%
  	%
  	%
  	%
  		In this single use case, every superchannel admits a realisation in terms of a quantum circuit with an encoder and a decoder \cite{chiribella08}. Let $\smap{C}$ be a superchannel transforming an input-operation $\map{\Lambda_\text{in}}:L(\H_2)\to L(\H_3)$ into $\smap{C}(\map{\Lambda_\text{in}})=\map{\Lambda_\text{out}}:L(\H_1)\to L(\H_4)$ and $\rho_\text{in}\in L(\H_1)$ be the input-state on which she would like to apply $\map{\Lambda_\text{out}}$. A protocol to implement the superchannel $\smap{C}$ can be realised as following:
\begin{enumerate}
\item  Alice performs an encoder channel $\map{E}:L(\H_1)\to L(\H_2\otimes \H_A)$ on the input-state $\rho_\text{in}\in L(\H_1)$.
\item  The input-operation $\map{\Lambda_\text{in}}:L(\H_2)\to L(\H_3)$ is performed on a  part of the state  $\map{E}(\rho_\text{in})\in L(\H_2\otimes \H_A)$.
\item The decoder $\map{D}:L(\H_3\otimes \H_A)\to L(\H_4)$ is applied to the state 
		$\left[ \map{\Lambda_\text{in}}\otimes \map{ \phantom{.}I_\text{A}	 }\right] \left(\map{E}(\rho_\text{in})\right)$ 
to obtain the final output-state
\begin{equation}
		\left[\smap{C}(\map{\Lambda_\text{in}})\right](\rho_\text{in})=\map{\Lambda_\text{out}}(\rho_\text{in}).
\end{equation}
\end{enumerate}
	In a delayed input-state protocol, the  encoder channel $\map{E}$ does not have access to the input-state $\rho_\text{in}$, since this state is only provided after the use of the operation $\map{\Lambda_\text{in}}$. Instead of having an encoder channel, Alice must then prepare a fixed state $\phi_E\in L(\H_2\otimes \H_A)$ that is independent of $\rho_\text{in}$. More precisely, a superchannel $\smap{C_D}$ represents a $k=1$ delayed input-state protocol if it can be realised by the following protocol:
\begin{enumerate}
\item  Alice prepares a state $\phi_E\in L(\H_2\otimes \H_A)$.
\item  The input-operation $\map{\Lambda_\text{in}}:L(\H_2)\to L(\H_3)$ is performed on a part of the state  $\phi_E\in L(\H_2\otimes \H_A)$ prepared by Alice.
\item The decoder $\map{D_D}:L(\H_1\otimes \H_3\otimes \H_A)\to L(\H_4)$ is applied to the state
 $\rho_\text{in}\otimes \left[ \left[\map{\Lambda_\text{in}}\otimes  \map{\phantom{.} I_\text{A} }\right] (\phi_E )\right]$
  to obtain the final output-state
\begin{equation}
		\left[\smap{C_D}(\map{\Lambda_\text{in}})\right](\rho_\text{in})=\map{\Lambda_\text{out}}(\rho_\text{in}).
\end{equation}
\end{enumerate}

  			We now consider parallel delayed input-state protocols with $k>1$ uses of the input-channel $\map{\Lambda_\text{in}}$. 
		 By definition, a parallel superchannel $\smap{C}: \left[ L(\mathcal{I})\to L(\mathcal{O})\right] \to \left[ L(\mathcal{I}_0)\to L(\mathcal{O}_{0})\right]$ that transforms $k$ identical input-operations into another can be represented by an encoder channel 	$\map{E}:L(\mathcal{I}_0)\to L(\mathcal{I} \otimes \mathcal{A})$ and a decoder channel $\map{D}:L(\mathcal{O} \otimes \mathcal{A})\to L(\O_0)$ such that
		 \begin{equation}
		 		\smap{C}(\map{\Lambda^{\otimes k}})= \map{D}\circ \left[ \map{\Lambda^{\otimes k}}\otimes \map{I}_\mathcal{A} \right]\circ \map{E}.
\end{equation}		
	 	That is, in order to perform the output-operation $  \map{\Lambda_\text{out}}=\smap{C}(\map{\Lambda^{\otimes k}})$ on an arbitrary input-state $\rho_{\text{\text{in}}}\in L(\mathcal{I}_0)$, we first perform the encoder operation on $\rho_\text{in}$, then the $k$ uses of $\map{\Lambda}$ on a part of the output of the encoder, and then the decoder $\map{D}$:	
	\begin{equation}
	\left[ \smap{C}\left( \map{\Lambda^{\otimes k}} \right) \right] (\rho_{\text{in}}) = %
	\map{D} \left( \left[ \map{\Lambda^{\otimes k}} \otimes \map{I_\mathcal{A}} \right] \left( \map{E} \left(\rho_{\text{in}}  \right) \right)\right).
 \end{equation}		

		In a delayed input-state protocol the encoder cannot not make use of the input-state $\rho_\text{in}$. Instead of an encoder channel $\map{E}$ we now consider some fixed (potentially entangled) quantum state $\phi_E\in L(\mathcal{I}\otimes \mathcal{A})$. On a delayed input-state protocol, the decoder $\map{D_D}:L(\mathcal{I}_0\otimes \O \otimes \mathcal{A})\to L(\O_0)$ acts directly on input-state $\rho_\text{in}$. We then say that a parallel superchannel $\smap{C_D}$ represents a delayed input-state parallel protocol if can be written as 
	\begin{equation} \label{eq:learn}
	\left[ \smap{C_D}\left( \map{\Lambda^{\otimes k}} \right) \right] (\rho_\text{in}) = %
	\map{D_D} \left(  \left[ \left[ \map{\Lambda^{\otimes k}} \otimes \map{I_\mathcal{A}} \right]  \left( \phi_E \right) \right]  \otimes \rho_\text{in}\right),
 \end{equation}	
 for some decoder channel $\map{D_D}:L(\mathcal{I}_0\otimes \O \otimes \mathcal{A})\to L(\O_0)$ and some state $\phi_E\in L (\mathcal{I}\otimes \mathcal{A})$.
 		If we define a $\psi_{\map{\Lambda^{\otimes k}}}:= \left[ \map{\Lambda^{\otimes k}} \otimes \map{I}_\mathcal{A} \right]   \left( \phi_E \right)$, we can re-rewrite
 	 Eq.\,\eqref{eq:learn} as 
 	\begin{equation}
	\left[ \smap{C_D}\left( \map{\Lambda^{\otimes k}} \right) \right] (\rho_\text{in}) = %
	\map{D}\left( \psi_{\map{\Lambda^{\otimes k}}}\otimes \rho \right).
 \end{equation}

 		 Parallel delayed input-state superchannels $\smap{C_D}$ also have a simple characterisation in terms its Choi operator $C_D\in L(\mathcal{I}_0 \otimes \mathcal{I} \otimes \O \otimes\O_0)$.
		 Since the encoder acts trivially on the space $L(\mathcal{I}_o)$, it follows from the same tools used to characterise standard ordered circuits \cite{chiribella07} that  $C_D$ represents a parallel delayed input-state protocol if and only 
  	\begin{equation} \label{eq:superchannel_PAR_delay} \begin{split} 
	C_D &\geq0 ;\\
	\tr_{\O_0} C_D &= \frac{I_{\mathcal{I}_0}} {d_{\mathcal{I}_0}} \otimes \tr_{\mathcal{I}_0 \O\O_0} C_D \otimes  \frac{I_{\mathcal{O}}}{d_{\mathcal{O}}}  ;  \\
	\tr(C_D) &= d_{\mathcal{I}_0} d_{\mathcal{O}}.  
	\end{split}  \end{equation}  
	Or equivalently, $C_D$ respects the standard parallel supermap restrictions of Eq.\,\eqref{eq:non_adaptive} and also 
	  	\begin{equation}  \begin{split} 
	\tr_{\O \O_0} C_D &= \tr_{\mathcal{I}_0 \O\O_0} C_D \otimes \frac{I_{\mathcal{I}_0}}{d_{\mathcal{I}_0}}.  
	\end{split}  \end{equation}  
	
		The formal definition and a simple Choi characterisation of adaptive  delayed input-state protocols follow straightforwardly from the discussions of the parallel case presented here%
  		\footnote{For adaptive  protocols where the input-operation $\map{\Lambda_\text{in}}$ can be used $k$ times one can also define the notion of $k$-delayed input-state protocol, where the input-state is provided after the $k$th use of the input-operation $\map{\Lambda_\text{in}}$. The characterisation of such protocols also follows from the discussion presented in this section and the methods presented in Sec.\,\ref{sec:intro}.}.
  		The case of  superchannels with indefinite causal order is more subtle. Since they have no encoder/decoder ordered quantum circuit implementation their physical interpretation is not evident. We let the precise definition and the characterisation of non-causally ordered delayed input-state protocols for future research.

		Probabilistic heralded parallel (adaptive) delayed input-state protocols are given by superinstruments whose elements add to a superchannel representing a  parallel (adaptive) delayed input-state protocol. It follows from the circuit realisation of quantum instruments \cite{bisio16} that every parallel (adaptive) delayed input-state protocol can be realised by an encoder ($k-1$ encoders) where the input-state is not required and a decoder, which makes use of the input-state, followed by a projective measurement.
				
				{We will now show that any probabilistic supermap can be implemented via a parallel probabilistic delayed input-state protocol with a smaller, but non-zero success probability.
				That is, if a supermap $\smap{S}$ represents a superinstrument element of some higher order transformation, there exists a delayed input-state parallel superinstrument which, when successful, implements the action of $\smap{S}$ in a probabilistic heralded way. This theorem holds true even if the supermap $\smap{S}$ corresponds to an indefinite causal order protocol.  Intuitively, one can undersand this theorem in terms of state teleportation and probabilistic heralded gate teleportation (see Sec.\,\ref{gate_teleportation} for a review of gate teleportation). In order to ``parallelise'' any superinstrument one can use the gate teleportation to re-arrange the position of all input-operations in parallel. Also, one can always delay the use of the input state by exploiting the state teleportation protocol \cite{bennett93}. Although the teleportation and gate teleportation protocol  may fail, the success probability is strictly positive for any fixed dimension, ensuring that the success probability of the parallel circuit is non-zero.}

				\begin{lemma} \label{parallel_is_OK}
Let $\smap{S}:[L(\bigotimes_{i=1}^{k}\mathcal{I}_i) \to L(\bigotimes_{o=1}^{k}\mathcal{O}_{o})]\to[L(\mathcal{I}_0)\to L(\mathcal{O}_{0}) ]$ be a supermap representing a general (possibly with indefinite causal order) probabilistic protocol that makes $k$ uses of a unitary operation $\map{U_d}$ and transforms to some other unitary operation $f(\map{U_d})$ with probability $p_U$ \ie, $\smap{S}\left(\map{U_d}^{\otimes k}\right)=p_U f(\map{U_d})$ .
There exists a parallel delayed input-state protocol implementing the supermap $\smap{S}$ with a probability greater than or equal to $\frac{p_U}{d_{total}}$, where $d_{total}$ is the product of all linear space dimensions, \ie, $d_{total}=\prod_{i=0}^k d_{\mathcal{I}_i}\ \prod_{i=0}^k d_{\O_i} $
\end{lemma}

\begin{proof}
	
		By assumption, $\smap{S} \left(\map{U_d^{\otimes k}}\right)= p_U f(\map{U_d})$ for all $U_d$ with probability $p_U$.
  Since $\smap{S}$ must be a superinstrument element, the corresponding Choi operator $S$ is positive and respects $\tr(S)\leq d_\O d_{\mathcal{I}_0}$, hence $ 0 \leq  \frac{1}{d_{total}} S\leq \frac{1}{d_\mathcal{I} d_{\mathcal{O}_0}} I$, where $I\in L(\mathcal{I}_0\otimes \mathcal{I}\otimes \O\otimes \O_0)$ is the identity operator. Note that the Choi operator $C_P:=\frac{1}{d_{\mathcal{I}} d_{\mathcal{O}_0}} I$ represents a valid parallel delayed input-state superchannel \ie, it satisfies the parallel delayed input-state superchannel conditions of Eq.\,\ref{eq:superchannel_PAR_delay}.  We thus define the new superinstrument via the Choi of its elements as $S_P:= \frac{1}{d_{total}} S$ and 
		$F_P:=\frac{1}{d_{\mathcal{I}} d_{\mathcal{O}_0}} I- \frac{1}{d_{total}} S$.
		 It follows that $F_P\geq 0$ and $S_P+F_P=C_P=\frac{1}{d_{\mathcal{I}} d_{\mathcal{O}_0}} I$ is a valid parallel delayed input-state superchannel, hence the operators $S_P$ and $F_P$ form a valid delayed-input state parallel superinstrument. By linearity, we can verify that $\smap{S_P} \left(\map{U_d^{\otimes k}}\right)=\frac{1}{d_{total}} \smap{S} \left(\map{U_d^{\otimes k}}\right)=\frac{p_U}{d_{total}}f(\map{U_d})$, ensuring that when the output associated to $S_P$ is obtained, the probabilistic parallel delayed input-state protocol represented by the superinstrument elements $S_P$ and $F_P$ performs the transformation of the supermap $\smap{S}$ with probability $\frac{p_U}{d_{total}}$.
\end{proof}


\section{Universal unitary complex conjugation} \label{sec:conjugation}

	In this section we consider the problem of transforming $k$ uses of an arbitrary $d$-dimensional unitary $\map{U_d}$ into its complex conjugate $\map{U_d^*}$ for some fixed basis. We prove that when $k<d-1$ uses are accessible, any exact unitary complex conjugation quantum protocol, including protocols with indefinite causal order, necessarily have zero success probability. In Ref.\,\cite{miyazaki17} the authors present a deterministic parallel quantum circuit that transforms $k=d-1$ uses of a $d$-dimensional unitary operation $\map{U_d}$ into its complex conjugate\footnote{  Reference \cite{miyazaki_thesis} also proves that when $d>2$, $k>1$ uses of the input-unitary operation $\map{U_d}$  are required for any non-null probabilistic heralded implementation.} $\map{U_d^*}$. Hence, when combined with Ref.\,\cite{miyazaki17}, our result reveals a characteristic threshold property for exact unitary complex conjugation: if $k<d-1$, universal exact unitary complex conjugation is impossible (zero success probability), if $k=d-1$ exact unitary complex conjugation is possible with probability one with a parallel circuit implementation.

\begin{theorem}[Unitary complex conjugation: no-go] \label{theo:conj_nogo}
		Any universal probabilistic heralded quantum protocol (including protocols without definite causal order) transforming $k<d-1$ uses of a $d$-dimensional unitary operation $\map{U_d}$ into its complex conjugate $\map{U_d^*}$ with probability $p$ that does not depend on $\map{U_d}$ necessarily has $p=0$, \ie, null success probability.
	
	\end{theorem}
	
\begin{proof}
		From Lemma\,\ref{parallel_is_OK} we see that if there exists a superinstrument that transforms $k$ uses of $\map{U_d}$ into its complex conjugate $\map{U_d^*}$ with some possibly smaller but still positive probability, there also exists a parallel superinstrument $\smap{S}$  that transforms $k$ uses of $\map{U_d}$ into its complex conjugate $\map{U_d^*}$ with some positive probability $p$, \ie, $\smap{S}(\map{U_d^{\otimes k}})= p\map{U_d^*}$. From the realisation theorem of the superinstruments (see Eq.\,\eqref{eq:realise_superis} and Ref.\,\cite{bisio10}), there exist an isometry $\map{E}:L(\mathcal{I}_0)\to L(\mathcal{I})\otimes L(\mathcal{A}) $  and an instrument element corresponding to success $\map{D_S}:L(\mathcal{I} \otimes \mathcal{A})\to L(\mathcal{O}_0)$ such that
 \begin{equation}
		\smap{S}( \map{U_d^{\otimes k}} )=  \map{D_S} \circ \left[ \map{U_d^{\otimes k}} \otimes   \map{I_A} \right] \circ \map{E}.
		\end{equation}
		Let $\rho_{\mathcal{IA}}$ 
		By the Naimark dilation, the instrument element $\map{D_S}$ is given by 
		\begin{equation} \label{DS}
		\map{D_S}(\rho_{\mathcal{IA}}) = \tr_{\mathcal{A}} \left( D \rho_{\mathcal{IA}} D^\dagger \right)
		\end{equation}
		for some operator $D\in L(\mathcal{I\otimes A})$.
		Set $\{\ket{a}\}$ as a basis for the auxiliary system $\mathcal{A}$.  The previous equation becomes
		\begin{equation} 
		\map{D_S}(\rho_{\mathcal{IA}}) = \sum_{a} \bra{a} D \rho_{\mathcal{IA}} D^\dagger \ket{a}.
		\end{equation}
		The operators  {$D_a:=\bra{a}D$} form a possible set of operators realizing the instrument $\map{D_S}$.
		 
		Since we assume $\smap{S} \left(\map{U_d^{\otimes k}}\right)= p \map{U_d'}$, $\smap{S} \left(\map{U_d^{\otimes k}}\right)$ must return a pure state whenever the input-state is a pure state.  The instrument element $\map{E}$ is an isometry, hence its output is always a pure state if the input-state is pure.  This forces $\map{D_S}$ to preserve the purity of pure input-states, which in turn implies that $\bra{a} D \rho_{\mathcal{IA}} D^\dagger \ket{a}$ must be the same for all $a$ up to a proportionality constant.   {Let $\map{D_a}$ denote the map given by
		 $\map{D_a}(\rho_{\mathcal{IA}} ):=D_a \rho_{\mathcal{IA}}  D_a^\dagger= \bra{a} D \rho_{\mathcal{IA}} D^\dagger \ket{a}$.  The above argument shows that $\map{D_a} \circ \left[ \map{U_d^{\otimes k}} \otimes   \map{I_A} \right] \circ \map{E}$ is also a valid universal conjugation supermap.}

		Without loss of generality we assume that $k=d-2$, since we may always opt to \textit{not} use any of the input-operations for the remaining cases of $k < d-2$.  The imaginary unit $\sqrt{-1}$ throughout this section will be denoted by the Roman font $\im$.  By hypothesis, every pure state $\ket{\psi}\in \mathcal{I}_0 \cong \mathbb{C}^d$ and unitary operator $U_d \in L(\mathbb{C})$ must respect
		\begin{equation} \label{action}
			D_a\left[ U_d^{\otimes d-2} \otimes I \right] E \ket{\psi} =e^{\im \phi_{\psi, U_d}}  \sqrt{p}U_d^*  \ket{\psi},
		\end{equation}
		where $\phi_{\psi, U_d}$ is a global phase that may depend on $\ket{\psi}$ and $U_d$.
		We see, however, that $\phi_{\psi, U_d}$ must be independent of the input-state.  Set $\{\ket{i}\}_{i=0}^{d-1}$ as the computational basis for $\mathbb{C}^d$ and the phase $\phi_{i, U_d}$ for when the input-state $\ket{\psi}$ is equal to $\ket{i}$.  Take a maximally entangled state $\ket{\phi_d^+} := \frac{1}{\sqrt{d}}\sum_{j=0}^{d-1} \ket{i}\ket{i}$ in $\mathcal{I}_0 \otimes \mathcal{I}_R$, where $\mathcal{I}_R$ is a ``copy'' of $\mathcal{I}_0$, i.e., another $d$-dimensional quantum system left untouched by $\smap{S}$.  We denote the corresponding phase by $\phi_{\phi^+_d,U_d}$.  Let $M_U := D_a\left[U_d^{\otimes d-2} \otimes I \right] E$.  Then, by linearity of $M_U$ and Eq.\,\eqref{action}, we conclude that $\phi_{i,U_d} = \phi_{\phi^+_d,U_d}$, hence no dependence on $i$.  The subscript of $\phi_{\psi,U_d}$ for the input-state shall be omitted as $\phi_{U_d}$
		 
		We now parametrise the operators $E$ and $D_a$ via their action on this basis as
		\begin{equation} \label{EDdef}
		\begin{split}
			E \ket{i}_{\mathcal{I}_0} \ &= \sum_{\vec{i},i,a}	\alpha_{\vec{i},i,a} \ket{i_1,\ldots,i_{d-2}}_\mathcal{I} \otimes \ket{a}_\mathcal{A}; \\
		_{\mathcal{O}_0}  \bra{i}D_a \ &= \sum_{\vec{i},i,a}	\beta_{\vec{i},i,a} \ket{i_1,\ldots,i_{d-2}}_\mathcal{I} \ \otimes \ket{a}_\mathcal{A} ,
		\end{split}
		\end{equation}
where $\vec{i} = [i_1,\ldots,i_{d-2}]$ is a vector such that $i_\lambda \in \{ 0,\ldots, d-1\}$ for any $\lambda = 1,\ldots,d-2$.
		Hereafter, we restrict to unitary operators $U_d$ that are diagonal in the computational basis such that $U_d=\sum_i e^{\im \theta_{i}} \ketbra{i}{i}$ where $\theta_{i}$ is any real number. For such diagonal $U_d$ its complex conjugate can be written as $U_d^*=\sum_i e^{-\im \theta_{i}} \ketbra{i}{i}$.  By Eq.\,\eqref{action},
		\begin{equation}
			\bra{i'}D_a\left[ U_d^{\otimes d-2} \otimes I \right]   E \ket{i'} =e^{\im \phi_{U_d}}  \sqrt{p} \bra{i'} U_d^*   \ket{i'}.
		\end{equation}
Substituting the definition \eqref{EDdef}, we obtain
		\begin{equation} \label{key0}
			\sum_{\vec{i},i,a} 	\alpha_{\vec{i},i,a} 	\beta_{\vec{i},i,a} \; e^{\im \left[\sum_{\lambda=1}^{d-2} \theta_{i_{\lambda}} \right]} = e^{\im \phi_{U_d}}  \sqrt{p} e^{-\im\theta_{i'}},
		\end{equation}
or, equivalently,
		\begin{equation} \label{key}
			\sum_{\vec{i},i,a} 	\alpha_{\vec{i},i,a} 	\beta_{\vec{i},i,a} \; e^{\im \left[\theta_{i'} + \sum_{\lambda=1}^{d-2} \theta_{i_{\lambda}}\right]} = e^{\im \phi_{U_d}}  \sqrt{p},
		\end{equation}
for all ${i,i'\in\{ 0,\ldots,d-1}\} $ and the diagonal $U_d$.  Note that each $U_d$ corresponds to some choice of real numbers $\vec{\theta_\lambda}= [\theta_0, \theta_1,\ldots,\theta_{d-1}]$ and \textit{vice versa}.  Moreover, the left-hand side of Eq.\,\eqref{key} depends on $i'$, but the right-had side does not.

		In combinatorics, a \textit{weak composition} of an integer $n$ is a sequence of non-negative integers that sum to $n$.  The weak compositions that appear in this proof are that of $d-1$ with $d$ elements.  The set of all such weak compositions will be denoted by $\Gamma$ and its elements (\ie,  the individual weak decomposition) by $\vec{\gamma}=[{\gamma}_0, \ldots, \gamma_{d-1}]$, where the subscripts denote the elements of $\vec{\gamma}$.
		
		 In Eq.\,\eqref{key} the summation on $i$ ranges between $0$ and $d-1$ and $\vec{i}$ over all possible combinations of $\vec{i}=[i_1,\ldots,i_{d-2}]$ where each $i_n$ ranges between $0$ and $d-1$.  Let $\nu_l$ denote the number of times an integer $l$ between $0$ and $d-1$ appears in $\vec{i}$ and $i$.  Recall that $\vec{i}$ consists of $d-2$ variables, thus $\vec{i}$ and $i$ in total are $d-1$ variables.  We see that $\sum_{l=0}^{d-1} \nu_l = d-1$.   Clearly, the sequence $[\nu_0,\ldots,\nu_{d-1}]$ belongs to $\Gamma$.  With slight abuse of notation, let us set $[\vec{i},i]=[i_1,\ldots,i_{d-2},i]$.  Each $[\vec{i},i]$ corresponds a $\vec{\gamma} \in \Gamma$.  Each $[\vec{i},i]$ with a given $\vec{\gamma}$ can be differentiated by an additional parameter, say $\kappa$. More specifically, let $K(\vec{\gamma})$ denote the set of all sequences $[\vec{i},i]$ with the weak decomposition $\vec{\gamma}$. This extra parameter $\kappa$ is then a natural number that  enumerates the sequences in $K(\vec{\gamma})$ (\eg, via lexicographic ordering).  Thus the summation $\sum_{\vec{i},i,a}$ in Eq.\,\eqref{key} can be relabelled as 
		$\sum_{\vec{\gamma},\kappa,a}$.		  Introducing $\alpha'_{\vec{\gamma}} := \sum_{\kappa,a} \alpha_{\vec{\gamma},\kappa,a}  \beta_{\vec{\gamma},\kappa,a}$, we have
		\begin{equation} \label{key2}
  			\sum_{\vec{\gamma}} 	\alpha_{\vec{\gamma}}' \; e^{\im \left[\sum_{l=0}^{d-1} \gamma_l \theta_{l}\right]} = e^{\im \phi_{U_d}}  \sqrt{p}.
		\end{equation}
Observe that for different $\vec{\gamma}$, the functions $e^{\im \left[ \sum_{l=0}^{d-1} \gamma_l \theta_{l}\right]}$ are linearly independent since $\theta_{i_\lambda}$ may take any value in the reals.  Each $\vec{\gamma}$ contains $d$ elements and must sum up to $d-1$.  One of the elements, say $\gamma_{l'}$ must be zero because the elements are non-negative.  Set $i' = l'$ in Eq.\,\eqref{key0} and use Eq.\,\eqref{key2} to replace $e^{i \phi_{U_d}} \sqrt{p}$ in the left-hand side of Eq.\,\eqref{key0}.  Then all the terms that appear in the right-hand side contain an exponent with a non zero coefficient in front of  $\theta_{l'}$, while the coefficients of $\theta_{l'}$ are zero on the left-hand side.   This equation can only be satisfied by setting $\alpha'_{\vec{\gamma}} = 0$, because $\exp(\im k \theta)$ and  $\exp(\im k' \theta)$ are linearly independent functions of $\theta$, for any pair of distinct integers $k$ and $k'$.  Thus, $p=0$.
\end{proof}


\section{Universal unitary transposition} \label{sec:trans}

		This section addresses the problem of universal unitary transposition. We consider probabilistic heralded exact universal quantum protocols transforming $k$ uses of a general $d$-dimensional unitary operation $\map{U_d}$ into its transpose $\map{U_d}$ in terms of a fixed basis. When only parallel protocols are considered, we show that the maximal success probability is exactly $p_\textit{s}=1-\frac{d^2-1}{k+d^2-1}$. Aslo, by exploiting ideas of the port-based teleportation \cite{ishizaka08}, one can design a delayed input-state parallel circuit that attains this maximal probability. When adaptive quantum circuits are considered, we present an explicit protocol that attains a success probability of 
			$p_\text{s}=1-\left(1- \frac{1}{d^2}\right)^{\ceil{\frac{k}{d}}}$,
		 which, for any constant dimension $d$, has an exponential improvement over any parallel protocol. We then analyse quantum protocols with indefinite causal order via the SDP approach presented in Sec.\,\ref{sec:SDP} and show that indefinite causal order protocols do have an advantage over causally ordered ones.
		


 \subsection{Gate teleportation and single-use unitary transposition} \label{gate_teleportation}
		Quantum teleportation is a universal protocol that can be used to send an arbitrary $d$-dimensional quantum state via classical communication assisted by quantum entanglement. We are going to describe the protocol for pure states, as the extension to general mixed states follows from linearity. Suppose Alice holds the qudit state $\ket{\psi}\in \mathbb{C}^d$ and shares with Bob a $d$-dimensional maximally entangled state 
		$\ket{\phi^+_d}:=\sum_{i=0}^{d-1} \frac{1}{\sqrt{d}}\ket{ii}$. In order to ``teleport'' her state to Bob, Alice performs a general Bell measurement on $\ket{\psi}$ and her share of the entangled state and then sends the outcome of her measurement to Bob. The generalised Bell measurements have POVM elements given by
	\small
	\begin{equation}  \begin{split}	
		\mathcal{M}:= \left\{  \left[ \left( X_d^i Z_d^j\right)^\dagger \otimes I_d \right] \ketbra{\phi^+_d}{\phi_d^+} \left[ \left( X_d{^i} Z_d{^j}\right) \otimes I_d  \right] \right\}_{i,j=0}^{i,j=d-1},
\end{split}  \end{equation}  
\normalsize		
		 where
	\begin{equation} \begin{split}	  \label{eq:shift}
	 X_d^i:=&\sum_{l=0}^{d-1} \ketbra{l\oplus i}{l} ; \\
	 Z_d^j:=&\sum_{l=0}^{d-1} \omega^{jl} \ketbra{l}{l},  
	\end{split}  \end{equation}  
	$\omega:= e^{\frac{2\pi \sqrt{-1}}{d}}$, and $l \oplus i$ denotes $l+i$ modulo $d$. The operators  $X_d^i$ and $Z_d^j$ are known as the shift and clock operators, respectively, and can be seen as a generalisation of the qubit Pauli operators.	Straightforward calculation shows that, after Alice's measurement, the state held by Bob is given by $X_d^i Z_d^j\ket{\psi}$. 
	
	After the measurement process is complete, Alice sends the measurement outcomes $i$ and $j$ of her joint measurement to Bob.  Bob can then apply the unitary operation  $(Z_d{^j})^{-1} (X_d{^i})^{-1} $  on his state to recover the state $\ket{\psi}$. Remark that, with probability $p=\frac{1}{d^2}$ Alice obtains the outcomes $i=j=0$ and Bob does not need to perform any correction.

    \begin{figure} 
	  \begin{center}
	\includegraphics[scale=.24]{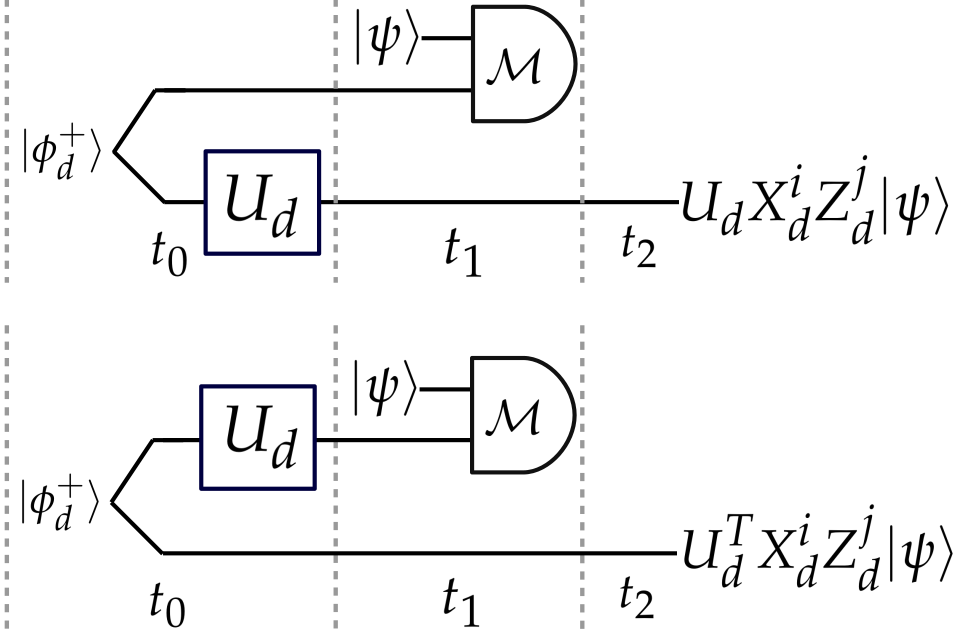} 
	\end{center}
			\caption{ Illustration of gate teleportation (upper circuit) and unitary transposition protocol (lower circuit).} \label{fig:teleportation}
	\end{figure}

	The standard teleportation protocol can be adapted to teleport the use of a unitary operation in a process known as gate teleportation \cite{gottesman99}. The idea here is that if Bob performs a unitary operation $U_d$ on his half of the maximally entangled state before Alice performs the joint Bell measurement, the final state is given by $U_d X_d^i Z_d^j\ket{\psi}$, 
	 see Fig.\,\ref{fig:teleportation}. In this protocol, the operation $U_d$ performed by Bob acts on the state $\ket{\psi}$ held by Alice when the outcomes are $i=j=0$, which happens with probability $p=\frac{1}{d^2}$. Gate teleportation can be represented as a quantum circuit (see Fig.\,\ref{fig:teleportation}) and has applications in fault tolerant quantum computation \cite{gottesman99}.
	
	Our method to transform a single use of a general {$d$-dimensional} unitary operation $\map{U_d}$ into its transpose%
	\footnote{The transposition is taken in the computational basis $\{\ket{i}\}_{i=0}^{d-1}$ in which the maximally entangled state $\ket{\phi^+_d}= \sum_i \frac{1}{\sqrt{d}}\ket{ii}$ is defined.}%
	 $\map{U_d^T}$ is based on the circuit interpretation of gate teleportation. The maximally entangled state respects the property ${I \otimes A \ket{\phi_d^+}= A^T \otimes I \ket{\phi_d^+}} $ for any linear operator ${A\in L(\mathbb{C}^d)}$. If Alice performs a general unitary $U_d$ on her half of the maximally entangled state, the state held by Bob after the protocol is $U^T_dX_d^i Z_d^j\ket{\psi}$. With probability $p=\frac{1}{d^2}$,  the outcome $i=j=0$ is obtained and $U^T_dX_d^i Z_d^j\ket{\psi}$ is equal to $U^T_d \ket{\psi}$, see Fig.\,\ref{fig:teleportation}.


\subsection{Port-based teleportation and parallel unitary transposition} \label{sec:PBT}

		Port-based teleportation \cite{ishizaka08} has the same main goal as the standard state teleportation protocol. Alice wants to ``teleport'' an arbitrary $d$-dimensional state $\ket{\psi}$  to Bob with classical communication assisted by shared entanglement. The original motivation of Port-based teleportation is to perform a teleportation protocol that does not require a correction made via
Pauli operators, but it can be made simply by selecting  some particular ``port''. For that, it allows more general initial resource
state and more general joint measurements. The three main differences of Port-based teleportation when compared to the standard
teleportation protocol presented in the previous section can be summarised by:
\begin{enumerate}
\item In port-based teleportation, instead of sharing a $d$-dimensional maximally entangled state, Alice and Bob may share a general $d^k$-dimensional entangled states $\ket{\phi}\in \left(\mathbb{C}^d \otimes \mathbb{C}^d\right)^{\otimes k }$. This general entangled state $\ket{\phi}$ can be seen as $k$ pairs of qudits, referred to as ``ports''.
\item Instead of performing a generalised Bell measurement, Alice can perform a general joint measurement on $\ket{\psi}$ and her half of the $k$ entangled states shared with Bob.
\item  Instead of performing the Pauli correction, Bob chooses a particular port based on Alice's message and discards the rest of the ports of his system.
\end{enumerate}	
		We note that since no Pauli correction is made, port-based teleportation can only perform the teleportation task approximately or probabilistically. In this paper we only consider the probabilistic exact port-based teleportation where Alice performs a $k+1$ outcome measurement, where $k$ outcomes are associated to the $k$ ports she shares with Bob and another outcome corresponding to failure. If Alice obtains the outcome of failure, she sends the failure flag to Bob and the protocol is aborted. If she obtains an outcome corresponding to some port $l$, she communicates this corresponding outcome to Bob and the state $\ket{\psi}$ is teleported to Bob's port labelled by $l$.
		
		The optimal probabilistic single port ($k=1$) case is essentially the standard state teleportation. Consider the case where Alice and Bob share the $d$-dimensional maximally entangled state $\ket{\phi_d^+}$. If we set the measurement performed by Alice as $M_1 = \ketbra{\phi_d}{\phi_d}$ and $M_{\text{fail}}= I - \ketbra{\phi_d}{\phi_d}$, with probability $p=\frac{1}{d^2}$, the state $\ket{\psi}$ is obtained in the single port $1$, and with probability $p_F=1-\frac{1}{d^2}$ the protocol fails.
		
		Reference\,\cite{studzinski16} shows that the optimal probabilistic port-based gate teleportation protocol for any dimension $d$ and number of states $k$ with success probability $p=1-\frac{d^2-1}{k+d^2-1}$. Reference\,\cite{studzinski16} also characterises the optimal $d^k$-dimensional shared entangled state and the optimal joint measurement Alice must perform. The optimal state resource state is described by exploiting the Schur-Weyl duality %
		\begin{equation}
		{\mathbb{C}^d}^{\otimes k}\cong \bigoplus_{\mu \in \text{irrep}(U^{\otimes k})} \mathbb{C}^{\text{dim}(\mu)}_\mu \otimes \mathbb{C}^{m_\mu},
		\end{equation}
		 where $\text{irrep}(U^{\otimes k})$ is the set of all irreducible representations $\mu$ of the group of special unitary $\text{SU}(U_d)$ contained in the decomposition $U^{\otimes k}$ and $m_\mu$ is the multiplicity of the representation $\mu$. The optimal resource state used for port-based teleportation can be written as
		\begin{equation}\label{eq:phi_pbt}
		 \ket{\phi_\text{PBT}}:=   \bigoplus_{\mu\in \text{irrep}(U^{\otimes k})} \sqrt{p_\mu} \ket{\phi^+(\mu)}\otimes \ket{\psi_{m_\mu}}, 
		\end{equation}
	where  
		\begin{equation}
					 \ket{\phi^+(\mu)}:=   \frac{1}{\sqrt{\text{dim}(\mu)}}\sum_i \ket{i_\mu i_\mu} \in \mathbb{C}^{\text{dim}(\mu)}_\mu \otimes \mathbb{C}^{\text{dim}(\mu)}_\mu  
		\end{equation}
		 is the maximally entangled state on the linear space of the irreducible representation $\mu$, $\{p_\mu\}$ is a probability distribution, and $\ket{\psi_{m_\mu}}\in \mathbb{C}^{m(\mu)}\otimes \mathbb{C}^{m(\mu)}$ is a pure quantum state.

		In Sec.\,\ref{gate_teleportation}  we have exploited the standard state gate teleportation to construct a protocol that can be used to transform a general unitary $U_d$ into its transpose $U_d^T$. We now exploit port-based gate teleportation to construct a parallel protocol that transforms $k$ uses of $U_d$ to obtain its transpose.
		
		    \begin{figure}  
	  \begin{center}
	\includegraphics[scale=.21]{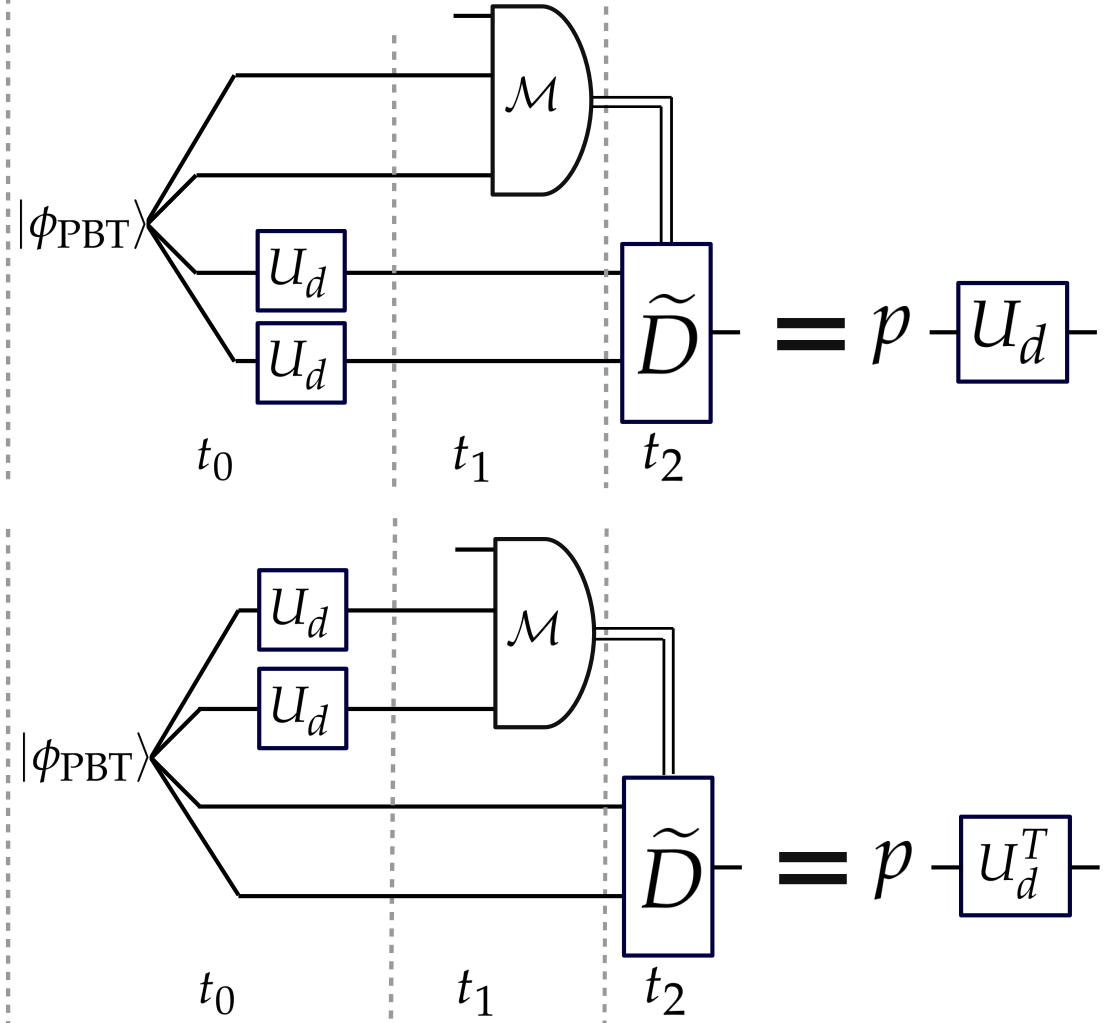} 
	\end{center}
			\caption{Illustration of the modified port-based teleportation protocol that makes $k=2$ uses of an arbitrary $d$-dimensional unitary operation $\map{U_d}$ where the state $\ket{\psi_\text{PBT}}$ is described in Eq.\,\eqref{eq:phi_pbt} and the decoder $\map{D}$ simply selects a particular port accordingly to the outcome of the joint measurement $\mathcal{M}$. The upper circuit exploits port-based gate teleportation to store $k=2$ uses of a input-operation $\map{U_d}$ and returns a single use of it with probability $p$. The lower circuit exploits port-based gate teleportation to transform $k$ uses of $\map{U_d}$ into a single use of its transpose $\map{U_d^T}$. The upper and lower circuits are successful with probability $p=1-\frac{d^2-1}{k+d^2-1}$.} \label{fig:PBT}
	\end{figure}

		The first important observation is that the state $\ket{\phi_\text{PBT}}$ (Eq.\,\eqref{eq:phi_pbt}) respects
		\begin{equation}		
			U_d^{\otimes k} \otimes I  \ket{\phi_\text{PBT}} = I \otimes  U_d^{T^{\otimes k}}  \ket{\phi_\text{PBT}}.
		\end{equation}
		This identity holds true because every tensor product of $k$ unitaries $U_d$ can be decomposed as\footnote{Here the symbol $\cong$ is used to ephasise that the Eq.\,\eqref{eq:SW} is true up to an isometry.}
		\begin{equation} \label{eq:SW}
		U_d^{\otimes k}\cong\bigoplus_{\mu\in \text{irrep}(U^{\otimes k})}   U(\mu) \otimes I_{m_{\mu}}
	\end{equation}	for some unitaries $U(\mu)$ acting on the irreducible representation space $\mathbb{C}^{\text{dim}(j)}_j$ \cite{studzinski16}. 	 Hence, similarly to the case of the single use unitary transposition, we can adapt port-based gate teleportation to obtain a general protocol to transform $k$ uses of a general unitary operation $\map{U_d}$ into its transpose $\map{U_d^T}$. It is enough to perform the operation $\map{U_d}$ on each of her half of entangled qudit states (see Fig.\,\ref{fig:PBT}). We will show in Sec.\,\ref{sec:trans_parallel} that this protocol is also optimal in terms of success probability.
		
\subsection{Review on probabilistic exact unitary learning}

		We make a brief summary of problem known as unitary learning (also known as storage and retrieval of unitary operations) \cite{vidal01,sasaki02,gammelmark09,bisio10,sedlak18}.
		{As we will show in Sec.\,\ref{sec:trans_parallel}, the problem of probabilistic unitary learning is closely connected to the problem of parallel unitary transposition and results related to unitary learning will be useful to prove the optimality of our parallel unitary transposition protocol.}
		Suppose that, until some time $t_1$, Alice has access to $k$ uses of some general $d$-dimensional unitary operation $U_d$ of which the description is not provided. 
		In a later moment $t_2$, where Alice cannot access $U_d$ any more, she wants to implement the action of this unitary on some general quantum state $\rho$ chosen at time $t_2$. 	   A parallel strategy%
		\footnote{In principle, one may also consider adaptive  protocols to perform better in the unitary learning problem. In an adaptive protocol, one can perform different enconder operations in between the use of the unitary to create more general protocols. One may also consider protocols where the unitaries $U_d$ are used without a definite causal order. 
			References\,\cite{bisio10,sedlak18} show that, for the unitary learning problem, the protocol with highest success probability (exact implementation) and highest expected fidelity (deterministic implementation) can always be parallelised.}
		 to succeed in this task is to perform the $k$ uses of $U_d$ on parts of an entangled quantum state $\phi_E$ before $t_1$ to obtain a quantum state $\psi_M:= \left[U_d^{\otimes k}\otimes I\right] \phi_E \left[U_d^{\dagger^\otimes k} \otimes I	 \right]  $. Alice then saves this state  $\psi_M$ until a later time $t_2$ where she performs a global decoder operation $\map{D}$ on the state $\psi_M$ together with the target state $\rho$, which is desired to satisfy%
		\footnote{We note that although the main goal is to obtain a decoder channel $\map{D}$ and entangled state $\phi_E$ such that %
		 ${\map{D}(\psi_M\otimes \rho)=U_d \rho U_d^\dagger}$ where 
		 ${\psi_M:= \left[U_d^{\otimes k}\otimes I \right] \phi_E \left[  U_d^{\dagger^\otimes k}\otimes I\right]  }$,
		  the unitary learning task cannot be realised in a deterministic and exact way for a general unitary $U_d$.
		} %
		 ${\map{D}(\psi_M\otimes \rho)=U_d \rho U_d^\dagger}$. References  \cite{vidal01,sasaki02,gammelmark09,bisio10} consider deterministic non-exact unitary learning protocols and analyse strategies that  simulate the action of $\map{U_d}$ with the maximal average fidelity, while Ref. \cite{sedlak18} considers probabilistic heralded protocols that can be used to retrieve (a single use of) $\map{U_d}$ exactly but may fail with some probability.

	The unitary learning problem described above can be rephrased as the problem of finding delayed input-state protocols that transform $k$ uses of a general unitary operation $\map{U_d}$ into itself. In Sec.\,\ref{sec:trans_parallel} we present a one-to-one connection between probabilistic unitary learning protocols and delayed input-state parallel protocols transforming $k$ uses of a general unitary operation $\map{U_d}$ into its transpose $\map{U_d^T}$. Essentially, we show that any probabilistic unitary learning with success probability $p$ can be translated into a parallel unitary transposition protocol with success probability $p$. This one-to-one connection is related to the fact that the optimal resource state used for unitary learning and the optimal resource state used for parallel delayed input-state unitary transposition can be both chosen as a state {$\ket{\phi}$ which respects the property}
			\begin{equation}		
			U_d^{\otimes k} \otimes I \ket{\phi} = I \otimes  U_d^{T^{\otimes k}}  \ket{\phi},
		\end{equation}
	as shown in next subsection.

		\subsection{Optimal parallel unitary transposition protocols} \label{sec:trans_parallel}

		We show how any parallel protocol that can be used to transform $k$ copies of  a general unitary operation $\map{U_d}$ into its transpose $\map{U_d^{T}}$ can be adapted into a delayed input-state protocol keeping the same success probability.

\begin{lemma} \label{theo:parallel_trans=delay}
		Any parallel probabilistic heralded protocol transforming $k$ copies of a general unitary $\map{U_d}$ into $\map{U_d^{T}}$ with a constant probability $p$ can be converted to a delayed input-state parallel protocol with the same probability $p$.
\end{lemma}
\begin{proof}

Let $S$ be the Choi operator of the superinstrument element associated to success and $F$ be the Choi operator of the superinstrument element associated to failure. Superinstrument element $S$ transforms $k$ copies of $\map{U_d}$ into $\map{U_d^{T}}$ with probability $p$, \ie, 
\begin{equation} \label{eq:S}
		\tr_{\mathcal{I}\O} \left( S \left[ I_{\mathcal{I}_0}\otimes  \C \left(\map{U_d^{\otimes k}} \right)^T \otimes I_{\O_0} \right] \right) = p \C(\map{U_d^{T}}) \quad \forall U_d,
\end{equation}
and $S+F$ is a valid parallel superchannel. 

Since $S$ transforms every unitary operator into its transpose, we can make the change of variable $U_d\mapsto BU_dA^T$ where $A$ and $B$ are arbitrary $d$-dimensional unitary operators. With that, unitary transposition can be seen as $\left( BU_dA^T \right)^{\otimes k}\mapsto p (BU_dA^T)^T = p A U_d^T B^T$. Our goal now is to show that if $S$ respects Eq.\,\eqref{eq:S}, any operator $S'$ respecting \begin{equation}  \begin{split} \label{eq:schur}
		S'= \Big[ A_{\mathcal{I}_0} \otimes  B_\mathcal{I}^{* \otimes k}  &\otimes A^{*\otimes k}_\O  \otimes B_{\O_0}\Big] \; S \\
		& \left[ A^\dagger_{\mathcal{I}_0}  \otimes B^{T \otimes k}_\mathcal{I} \otimes A^{T \otimes k}_\O \otimes B^\dagger_{\O_0}  \right]
	\end{split}  \end{equation}  
	satisfies 
	\begin{equation} \label{ABBA}
		\tr_{\mathcal{I}\O} \left( S' \left[ I_{\mathcal{I}_0} \otimes {\C\left( \map{U_d^{\otimes k}} \right)}^T \otimes I_{\O_0} \right] \right) = p \C (\map{U_d^{T}}) \quad \forall U_d.
\end{equation}

	To prove this fact, first note that the identity presented in Eq.\,\eqref{eq:idenitity_choi} implies that ${\C(\map{AU_d^{T}B^T})= \left[ A\otimes B \right] \,\C(\map{U_d^{T}})  \left[A^\dagger\otimes B^\dagger\right]}$, and
	\small
	\begin{equation}
	 \C \left(  \left[\map{B U_d A^T} \right]^{\otimes k} \right) = \left[ B^{\otimes k} \otimes A^{\otimes k}\right]  \C (\map{U_d^{\otimes k}})  \left[ B^{\dagger^{\otimes k}}\otimes A^{\dagger^{\otimes k}} \right],
	 	\end{equation}
	 		\normalsize
	 	which implies
	 	\small
	 		\begin{equation} \label{eq:sym_trans}
	 \C \left(  \left[\map{B U_d A^T }\right]^{\otimes k} \right)^T = \left[ B^{*^{\otimes k}} \otimes A^{*^{\otimes k}} \right]  \C (\map{U_d^{\otimes k}})^T  \left[ B^{T^{\otimes k}}\otimes A^{T^{\otimes k}} \right].
	 	\end{equation}
	 	\normalsize
	 	Substituting Eq.\,\eqref{eq:sym_trans} and ${\C(\map{AU_d^{T}B^T})= \left[ A\otimes B \right] \,\C(\map{U_d^{T}})  \left[A^\dagger\otimes B^\dagger\right]}$ in Eq.\,\eqref{eq:S} we obtain
	 	\footnotesize%
	 	\begin{equation} \label{eq:final}  \begin{split}
		&\tr_{\mathcal{I}\O} \left( S \left[ I_{\mathcal{I}_0}\otimes \left[ B^{*^{\otimes k}} \otimes A^{*^{\otimes k}} \right]  \C (\map{U_d^{\otimes k}})^T  \left[ B^{T^{\otimes k}}\otimes A^{T^{\otimes k}} \right] \otimes I_{\O_0} \right] \right) \\
		&= \left[A_{\mathcal{I}_0}\otimes B_{\mathcal{O}_0}\right]\,\C(\map{U_d^{T}}) \left[A_{\mathcal{I}_0}^\dagger\otimes B_{\mathcal{O}_0}^\dagger\right] .
	\end{split}  \end{equation}
	\normalsize
 	If we apply the operator $A^\dagger_{\mathcal{I}_0} \otimes B^\dagger_{\mathcal{O}_0}$ on the left side and the operator $A_{\mathcal{I}_0} \otimes B_{\mathcal{O}_0}$ on the right side of Eq.\eqref{eq:final} and use the cyclic property of the trace, we find that $S$ can be substituted by
 	 \begin{equation}  \begin{split} 
		S'':= \Big[ A^\dagger_{\mathcal{I}_0} \otimes  B_\mathcal{I}^{T \otimes k}  &\otimes A^{T\otimes k}_\O  \otimes B^\dagger_{\O_0}\Big] \; S \\
		& \left[ A_{\mathcal{I}_0}  \otimes B^{* \otimes k}_\mathcal{I} \otimes A^{* \otimes k}_\O \otimes B_{\O_0}  \right].
	\end{split}  \end{equation}  
 	Since $A$ and $B$ are arbitrary unitary operators, we can take the invertible transformations $A^\dagger\mapsto A$ and $B^\dagger\mapsto B$ to obtain the symmetry of Eq.\,\eqref{eq:schur}.

		The symmetry presented in Eq.\,\eqref{eq:schur} motivates the definition of a Haar measure ``twirled'' map
	\begin{equation}  \begin{split}
	\map{\tau}(S):= \int_{\text{Haar}} & \left[ A_{\mathcal{I}_0}  \otimes B^{*\otimes k}_\mathcal{I} \otimes A^{*\otimes k}_\O \otimes B_{\O_0} \right] S \\
	& \left[ A^\dagger_{\mathcal{I}_0}\otimes B^{T \otimes k}_\mathcal{I} \otimes A^{T\otimes k}_\O   \otimes  B^\dagger_{\O_0}  \right] \text{d}{A}\text{d}B.  
	\end{split}  \end{equation}  
	We now define a twirled version of the superinstrument as $S_\tau:=\map{\tau}(S)$ and $F_\tau:=\map{\tau}(F)$, which respects the conditions of valid superinstruments and 
	$S_\tau$ also transforms $k$ uses of any $\map{U_d}$ into $\map{U_d^{T}}$ with probability $p$. We now notice that both $S_\tau$ and $F_\tau$ respects
	\begin{equation}  \begin{split}
	&\tr_{\O \O_0} S_\tau = \\
	& \int_{\text{Haar}}  \left[ A_{\mathcal{I}_0}\otimes B^{* \otimes k}_\mathcal{I} \right] \tr_{\O\O_0}(S) \ \left[ A^\dagger_{\mathcal{I}_0} \otimes B^{T \otimes k}_\mathcal{I} \right] \text{d}{A}\text{d}B  \\
	&\propto I_{\mathcal{I}_0} \otimes \tr_{\mathcal{I}_0 \O \O_0} S_\tau \,
	\end{split}  \end{equation}  
	since the identity is the only operator that commutes with all unitary operations (Schur's lemma). It follows then that the superchannel $C_\tau:= S_\tau+F_\tau$ respects the conditions of a parallel  delayed input-state protocol.
\end{proof}

\begin{lemma} \label{theo:trans=learning}
For every delayed input-state parallel protocol transforming $k$ uses of a general unitary operation $\map{U_d}$ into its transpose $\map{U_d^T}$ with success probability $p$ that is independent of $\map{U_d}$, there exists a probabilistic unitary learning protocol with a success probability $p$.

Conversely, for every probabilistic unitary learning protocol with a success with probability $p$  that is independent of $\map{U_d}$, there exists a delayed input-state parallel protocol transforming $k$ uses of a general unitary operation $\map{U_d}$ into its transpose $\map{U_d^T}$ with a constant success probability $p$.
\end{lemma}
\begin{proof}
We start by showing how one can adapt a parallel protocol transforming $k$ uses of a general unitary operation $\map{U_d}$ into its transpose $\map{U_d^T}$ into a unitary learning one with the same success probability.  

	Let $S$ be the Choi operator of the superinstrument element associated to success and $F$ be the Choi operator of the superinstrument element associated to failure. Superinstrument element $S$ transforms $k$ copies of $\map{U_d}$ into $\map{U_d^{T}}$ with probability $p$, \ie, 
	\begin{equation} \label{eq:S2}
		\tr_{\mathcal{I}\O} \left( S \left[ I_{\mathcal{I}_0}\otimes  \C \left(\map{U_d^{\otimes k}} \right)^T \otimes I_{\O_0} \right] \right) = p \C(\map{U_d^{T}}) \quad \forall U_d,
\end{equation}
Lemma \ref{theo:parallel_trans=delay} states that this protocol can be converted to have a delayed input-state and without lost of generality, the superchannel $C=S+F$ respects the commutation relation
	\begin{equation} \label{eq:schur_trans}
	 \left[C,A^*_{\mathcal{I}_0} \otimes B^{\otimes k}_\mathcal{I} \otimes A^{\otimes k}_\O\otimes B^*_{\O_0}\right]=0 
\end{equation}	 for every pair of unitary operations $A,B\in SU(d)$.

	When a Choi operator $C$ represents a delayed input-state protocol, the operator $C_\mathcal{I}:= \tr_{\mathcal{I}_0\O\O_0} C$ is proportional to the reduced state $\tr_\mathcal{A} \left( \phi_E\right)$ of the state $\phi_E\in L(\mathcal{I}\otimes \mathcal{A})$ prepared by Alice before the use of the input-operations%
	\footnote{See Fig.\,\ref{fig:delayed} for a pictorial illustration for the case $k=1$. Let  $\phi_{E} \in L(\mathcal{H}_2\otimes \H_A) $ be the state created by the encoder of the delayed input-state protocol of Fig.\,\ref{fig:delayed}. In this case, $C_2:=\tr_{134}C$ is proportional the reduced state $\tr_\mathcal{A} \phi_{E}$.}. 
	From the commutation relation in Eq.\,\eqref{eq:schur_trans},
	 we see that  $C_\mathcal{I}$ respects
		\begin{equation} \label{eq:schurB}
	 \left[C_\mathcal{I}, B^{\otimes k}_\mathcal{I} \right]=0. 
\end{equation}
%
		 
		 The Schur-Weyl duality states that $k$ identical $d$-dimensional unitaries $B$ can be decomposed as (see Sec.\,\ref{sec:PBT})
\begin{equation}
 B^{\otimes k} \cong \bigoplus_{\mu \in \text{irrep}\left(U_d^{\otimes k}\right) } B(\mu) \otimes I_{m(\mu)},
\end{equation}
where $B(\mu)\in L \left( \mathbb{C}^{\text{dim}(\mu)}_\mu \right)$ is a unitary operator, and $I_{m(\mu)}$ is the identity on the multiplicity space $\mathbb{C}^{m(\mu)}$.
	Since the reduced state $\tr_\mathcal{A}(\phi_E)$ respects the relation $\left[ \tr_\mathcal{A} \left(\phi_E\right) ,B^{\otimes k} \right]=0$, Schur's lemma ensures that the reduced encoder state has the form of 
	\begin{equation} \label{eq:reduced_phi}
	 \tr_\mathcal{A} \left(\phi_E\right) \propto  \bigoplus_\mu I_\mu \otimes \rho_{m_\mu},
	\end{equation}
	where $I_\mu$ is the identity on the the linear space $\mathbb{C}^{\text{dim}(\mu)}_\mu $  and $\rho_{m_\mu}$ is some state on the multiplicity space of $\mu$. 
	Without loss of generality, we can assume that $\phi_E=\ketbra{\phi_E}{\phi_E}$ is a pure state with a reduced state that respects Eq.\,\eqref{eq:reduced_phi}. It follows then that $\ket{\phi_E}$ can be written as 
\begin{equation} 
		 \ket{\phi_E}:= \bigoplus_{\mu\in \text{irrep}(U^{\otimes k})} \sqrt{p_\mu} \ket{\phi^+(\mu)}\otimes \ket{\psi_{m_\mu}}, 
\end{equation}		
	where  
		\begin{equation}
			\ket{\phi^+(\mu)}:=  \frac{1}{\sqrt{\text{dim}(\mu)}}\sum_i \ket{i_\mu i_\mu} \in \mathbb{C}^{\text{dim}(\mu)}_\mu \otimes \mathbb{C}^{\text{dim}(\mu)}_\mu  
		\end{equation}
		 is the maximally entangled state on the linear space of the irreducible representation $\mu$, $\{p_\mu\}$ is a probability distribution, and $\ket{\psi_{m_\mu}}\in \mathbb{C}^{m(\mu)}\otimes \mathbb{C}^{m(\mu)}$ are some purifications of $\rho_{m_\mu}$.

		We now make an important observation. Although the state $\ket{\phi_E}$ is not the maximally entangled state, it respects
		\begin{equation}		
			U_d^{\otimes k} \otimes I \ket{\phi_E} = I \otimes  U_d^{T^{\otimes k}}  \ket{\phi_E}.
		\end{equation}
		This identity holds true because any tensor product of $k$ identical unitaries $U_d$ can be decomposed as 
		\begin{equation}
		U_d^{\otimes k}\cong \bigoplus_{\mu\in \text{irrep}(U^{\otimes k})}   U(\mu) \otimes I_{m_{\mu}}
	\end{equation}	for some unitaries $U(\mu)$ acting on the invariant representation space $\mathbb{C}^{\text{dim}(j)}_j$. Any delayed input-state protocol that can be used for unitary transposition can be used for unitary learning, since it is enough to perform the unitaries $U_d^{\otimes k}$ on the ``other'' half of the entangled state $\ket{\phi_E}$ on which the joint operation is not performed.

		We now show how to transform probabilistic unitary learning protocols to heralded unitary transposition protocols.
		In Ref.\,\cite{sedlak18}, the authors have shown that, without loss of generality, any probabilistic unitary learning protocol can be made parallel and, moreover, with the entangled state $\ket{\phi_E}\in L(\mathcal{I}\otimes \mathcal{A})$ which respects the property
		\begin{equation}
			U_d^{\otimes k} \otimes I \ket{\phi_E} = I \otimes  U_d^{T^{\otimes k}}  \ket{\phi_E}.
		\end{equation}
		Hence, if we perform the unitary operations $U_d^{\otimes k}$ into the half of the entangled state $\ket{\psi}$ on which the joint measurement is performed, the unitary recovered after the learning protocol will be $U_d^T$ instead of $U_d$.
\end{proof}

 We are now in position to prove that the protocol based on port-based gate teleportation presented in Sec.\,\ref{sec:PBT} is optimal.
	
\begin{theorem}[Optimal parallel unitary transposition] \label{theo:trans_parallel}
		The modified port-based gate teleportation protocol can be used to transform $k$ uses of an arbitrary $d$-dimensional unitary operation $\map{U_d}$ into its transpose $\map{U_d^T}$ with success probability $p=1-\frac{d^2-1}{k+d^2-1}$ in a parallel delayed input-state protocol.
Moreover, this protocol attains the optimal success probability among all parallel protocols with probability $p$ that does not depend on $\map{U_d}$.
\end{theorem}

\begin{proof}
		As shown above, the identity $U^{\otimes k} \otimes I  \ket{\phi_\text{PBT}} = I \otimes  U^{T^{\otimes k}}  \ket{\phi_\text{PBT}}$ ensures that port-based gate teleportation can be used to construct a delayed input-state parallel protocol that obtains $U_d^T$ with $k$ uses of $U_d$ with probability  $p=1-\frac{d^2-1}{k+d^2-1}$. 

		Lemma\,\ref{theo:trans=learning} shows that any protocol transforming $k$ uses of $\map{U_d}$ into its transpose $\map{U_d^T}$ in a parallel protocol with probability $p$ can be used to succesfully ``learn'' the input-operation $\map{U_d}$ with probability $p$ and $k$ uses. Reference\,\cite{sedlak18} shows that the optimal protocol for unitary learning a unitary $U_d$ with $k$ uses cannot have constant probability greater than $p=1-\frac{d^2-1}{k+d^2-1}$, which bounds our maximal probability of success and finishes the proof.
		
\end{proof}


	\subsection{Adaptive unitary transposition protocols} \label{sec:trans_sequential}

		In this subsection we present an adaptive circuit that transforms $k$ uses of an arbitrary $d$-dimensional unitary operation $\map{U_d}$ into a single use of its transpose $\map{U_d^T}$ with heralded probability ${p=1-\left(1-\frac{1}{d^2}\right)^{\ceil{\frac{k}{d}}}}$ (see Fig.\,\ref{fig:adaptive_trans}).

  \begin{figure} 
	  \begin{center}
	\includegraphics[scale=.24]{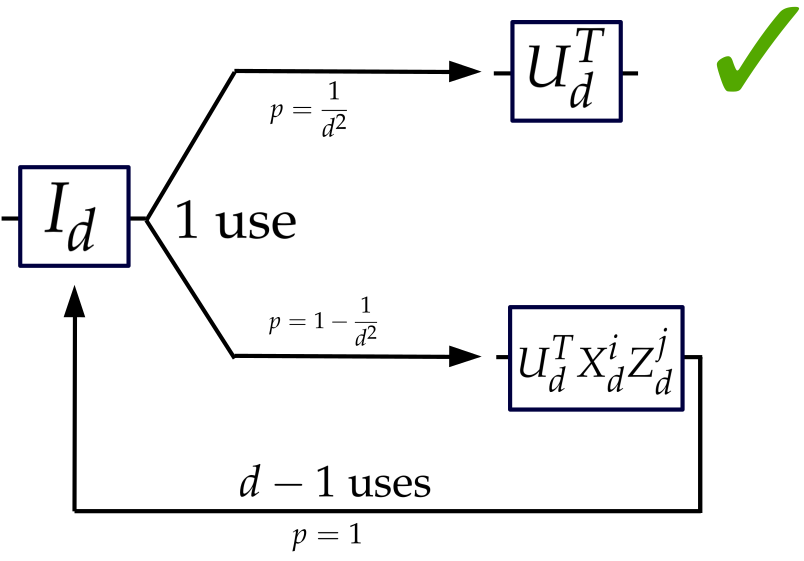} 
	\end{center}
			\caption{A flowchart illustrating the adaptive unitary transpose protocol.} \label{fig:adaptive_trans}
	\end{figure}

\begin{enumerate}
\item We start by making a single use of the input-operation $\map{U_d}$ to implement the probabilistic heralded transposition protocol based on gate teleportation described in Sec.\,\ref{sec:trans_parallel}. When the generalised Bell measurement returns the outcomes $i$ and $j$, the operator $U_d$ is transformed into $V_1=U_d^T X^i_d Z^j_d$, where $X^i_d $ and $Z^j_d $ are the clock and shift operators, respectively (see Eq.\,\eqref{eq:shift}).
\item If both outcomes $i$ and $j$ correspond to the identity operator, \ie, $i=j=0$, we have $V_1=U_d^T$ and we stop the protocol with success. If some other outcome is obtained, we make $d-1$ uses of $U_d$ to implement the unitary complex conjugate protocol  \cite{miyazaki17}  to obtain $U_d^*$. We then apply ${X_d^i}^{-1}{Z_d^j}^{-1} U_d^*$ into $V_1$ to ``cancel'' the transformation of step 1 to obtain identity operator
$\left[ {X_d^i}^{-1}{Z_d^j}^{-1} U_d^*   U_d^T Z_d^j X_d^i=I_d \right]$.
\item Go to step 1. 
\end{enumerate}

We see that step 1 fails returning $\map{U^T_d}$ with probability $\left(1-\frac{1}{d^2}\right)$ and we need in total $d$ uses of the input-operation $\map{U_d}$ to complete steps 1 and 2. These steps may be repeated up to $\ceil{\frac{k}{d}}$ times, hence they lead to a success probability of ${p=1-\left(1-\frac{1}{d^2}\right)^{\ceil{\frac{k}{d}}}}$.

\subsection{Optimal protocols via SDP formulation and indefinite causal order advantage} \label{sec:SDP_trans}

	We apply the SDP methods obtained in Sec.\,\ref{sec:SDP} to the case of unitary transposition and present the optimal success probability in Table\,\ref{table:SDP2}. By checking Table\,\ref{table:SDP2} one observes that the adaptive circuit we have presented in Sec.\,\ref{sec:trans_sequential} is not optimal. {One possible intuitive understanding is that the adaptive protocol we have presented in Sec.\,\ref{sec:trans_sequential} ``wastes'' $d-1$ uses of the input-operation $\map{U_d}$ to recover the input-state.} We also notice that indefinite causal order protocols provide a strictly large success probability when compared to causally ordered ones. It is interesting to observe that although indefinite causal order protocols have been reported useful in tasks such as non-signalling channel discrimination \cite{chiribella11}, quantum computation \cite{araujo14}, and quantum channel capacity activation \cite{ebler18,salek18}, this is the first time that indefinite causal order protocols outperform causally ordered ones when multiple uses of the same unitary input-operation are made.  In those previous examples cited, the advantage of indefinite causal order was obtained by exploiting the \textit{quantum switch} \cite{chiribella09}, a process which is not useful in our task of unitary channel transformation, since the quantum switch would transform $k$ uses of the any unitary operation $\map{U_d}$ into  simple $k$ concatenations of $\map{U_d}$, or equivalently, a single use of $\map{U_d^k}$. Our results for indefinite causal order then reveals the existence of a different class of indefinite causal order protocols, similarly to the one reported for unitary inverse in Ref.\,\cite{PRL}.

	\begin{table}[t!]
\begin{tabular}{|c|c|c|c|}
\hline 
$d=2$ & Parallel & Adaptive & Indefinite causal order\\
\hline 
$k=1$ & \blue{$\frac{1}{4}=0.25$}  &\blue{$\frac{1}{4}=0.25$}  & \blue{$\frac{1}{4}=0.25$}   \\ 
\hline 
$k=2$ & \blue{$\frac{2}{5}=0.4$}  & $ 0.4286\approx \frac{3}{7}$ & $0.4444\approx \frac{4}{9}$  \\ 
\hline 
$k=3$ & \blue{$\frac{1}{2}=0.5$}  & $0.7500\approx\frac{3}{4}$ & $0.9416$  \\ 
\hline 
%
%
\hline 
$d=3$ & Parallel & Adaptive & Indefinite causal order\\
\hline 
$k=1$ & \blue{$\frac{1}{9}\approx0.1111$}  & \blue{$\frac{1}{9}\approx0.1111$} & \blue{$\frac{1}{9}\approx0.1111$}   \\ 
\hline 
$k=2$ & \blue{$\frac{2}{10}=0.2$}  & $0.2222\approx \frac{2}{9}$ & $0.2500\approx \frac{2}{8}$ \\ 
\hline 
\end{tabular} 
\caption{Table with optimal success probability we have obtained for heralded protocols transforming $k$ uses of $\map{U_d}$ into a single use of its transpose $\map{U^{T}_d}$. The values in blue were proved analytically and the values in black were obtained via numerial SDP optimisation.}\label{table:SDP2}
\end{table}


\section{Universal unitary inversion protocols} \label{sec:inverse}

	We now address the problem of transforming $k$ uses of a general $d$-dimensional unitary operation $\map{U_d}$ into a single use of its inverse $\map{U_d^{-1}}$ with probabilistic heralded quantum circuits. We have presented our adaptive circuit in Ref.\,\cite{PRL} and here we present a parallel implementation and provide more details on the adaptive circuit.
	
	Before presenting our protocols we prove that, similarly to the complex conjugation case, any protocol performing exact universal unitary inversion with $k<d-1$ uses of the unitary input-operation $\map{U_d}$ necessarily has null success probability. Also, this no-go result also holds even when protocols with indefinite causal order are considered.
		\begin{theorem}[Unitary inversion: no-go]\label{theo:inv_nogo}
		Any universal probabilistic heralded quantum protocol (including protocols without definite causal order) transforming $k<d-1$ uses of a $d$-dimension unitary operation $\map{U_d}$ into a single use of its 				inverse $\map{U_d^{-1}}$ with success probability $p$ that does not depend on $U_d$ necessarily has $p=0$ \ie, null success probability.
	\end{theorem}
	
	\begin{proof}
		Assume that there exists a quantum protocol transforming $k$ uses of a general $d$-dimensional unitary operation $\map{U_d}$ into its inverse $\map{U_d^{-1}}$ with a non-zero success probability $p$. We can then exploit the single-use unitary transposition protocol presented in Sec.\,\ref{gate_teleportation} to obtain $\map{U_d^*}$ with success probability $p/d^2>0$, which contradicts Lemma\,\ref{theo:conj_nogo}.
	\end{proof}

\subsection{Parallel unitary inversion protocols}

We start by showing that, similarly to universal parallel transposition, any universal parallel unitary inversion protocol can be made in a delayed input-state way.

\begin{lemma} \label{theo:learning}
Any parallel probabilistic heralded parallel protocol $k$ uses of a general unitary $\map{U_d}$ into a single use of its inverse $\map{U_d^{-1}}$ with constant probability $p$ can be conversed to a delayed input-state parallel protocol with the same probability $p$.
\end{lemma}
\begin{proof}
The proof follows the same steps as the one in Theorem \ref{theo:parallel_trans=delay}. The only difference is that for unitary transposition, the superinstrument element $S$ can be chosen as an operator that commutes with unitaries of the form 
			$ A_{\mathcal{I}_0}  \otimes B^{*\otimes k}_\mathcal{I} \otimes A^{*\otimes k}_\O \otimes B_{\O_0}$ and for unitary inversion $S$ can be chosen as an operator which commutes with all unitaries of the form of 
			 $A_{\mathcal{I}_0} \otimes B^{\otimes k}_\mathcal{I} \otimes A^{\otimes k}_\O\otimes B_{\O_0} $.
\end{proof}

	We are now in conditions to present a universal circuit for parallel unitary inversion and also to obtain an upper bound on the maximal success probability. Our protocol makes use of the unitary complex conjugation and unitary transposition and it is proven to be optimal for qubits.

\begin{theorem} [Universal unitary inverse]
	There exists a parallel delayed input-state probabilistic quantum circuit that transforms $k$ uses of an arbitrary $d$-dimensional unitary operation $\map{U_d}$ into a single use of its inverse $\map{U_d^{-1}}$ with success probability $p_\text{S} = 1- {\frac{d^2-1}{k'+d^2-1}}$ where $k':=\floor{\frac{k}{d-1}}$ is the greatest integer that is less than or equal to $\frac{k}{d-1}$. 
	
The maximal success probability transforming $k$ uses of an arbitrary $d$-dimensional unitary operation $\map{U_d}$ into a single use of its inverse $\map{U_d^{-1}}$ in a parallel quantum circuit is upper bounded by ${p_\text{max}\leq 1-\frac{d^2-1}{k(d-1)+d^2-1}}$.
\end{theorem}

\begin{proof}

		We construct our protocol by concatenating the protocol for unitary complex conjugation of Ref.\,\cite{miyazaki17} with the unitary transposition one presented in Sec.\,\ref{sec:PBT}. First we divide the $k$ uses of the input-operation $\map{U_d}$ into $k'=\floor*{\frac{k}{d-1}}$ groups containing $d-1$ uses of $\map{U_d}$ and discard possible extra uses. We then exploit the unitary conjugation protocol to obtain $k'$ uses of $U_d^*$. After, we implement the unitary transposition protocol of Sec.\,\ref{sec:PBT} on $k'$ uses of $\map{U_d^*}$ to obtain a single use of $\map{U^{-1}_d}$ with probability of success given by
	$p=1-\frac{d^2-1}{k'+d^2-1}$.
	
	Next, we prove the upper bound. Let $p_\text{inv}(d,k)$ be the success probability of transforming $k$ uses of an arbitrary unitary input-operation $\map{U_d}$ into a single use of its inverse $\map{U_d^{-1}}$  with a parallel circuit. Suppose one has access to $l=k(d-1)$ uses of an input-operation $\map{U_d}$. One possible protocol to transform these $l$ uses of $\map{U_d}$ into its transpose with a parallel circuit is the following, first we perform the deterministic parallel complex conjugation protocol on $l$ uses to obtain $k$ uses of  $\map{U_d^{*}}$. We then perform the parallel unitary inversion on $k$ uses of  $\map{U_d^{*}}$ to obtain  $\map{U_d^{*^{-1}}}=\map{U_d^T}$ with probability $p_\text{inv}(d,k)$. This parallel unitary transposition protocol has then success probability of  $q_\text{T}(d,l)=p_\text{inv}(d,k)$. Theorem\,\ref{theo:trans_parallel} states that any parallel circuit that transforms $l$ uses of an arbitrary unitary into its transpose respects $q_\text{T}(d,l)\leq 1-\frac{d^2-1}{l+d^2-1}$, which implies 
	\begin{equation}
	 p_\text{inv}(d,k)\leq 1-\frac{d^2-1}{l+d^2-1}=1-\frac{d^2-1}{k(d-1)+d^2-1},
	 \end{equation}	
	  what completes the proof.

\end{proof}

\subsection{Adaptive unitary inversion circuit} \label{sec:inv_sequential}

		For completeness we now summarise the protocol for adaptive unitary inversion presented in Ref.\,\cite{PRL}. 
		Our protocol to obtain $\map{U_d^{-1}}$ follows similar steps of the protocol to implement $\map{U^{-1}_d}$ presented in the previous section goes as follow (see Fig.\,\ref{fig:inv_sequential}).
  \begin{figure}  
	  \begin{center}
	\includegraphics[scale=.24]{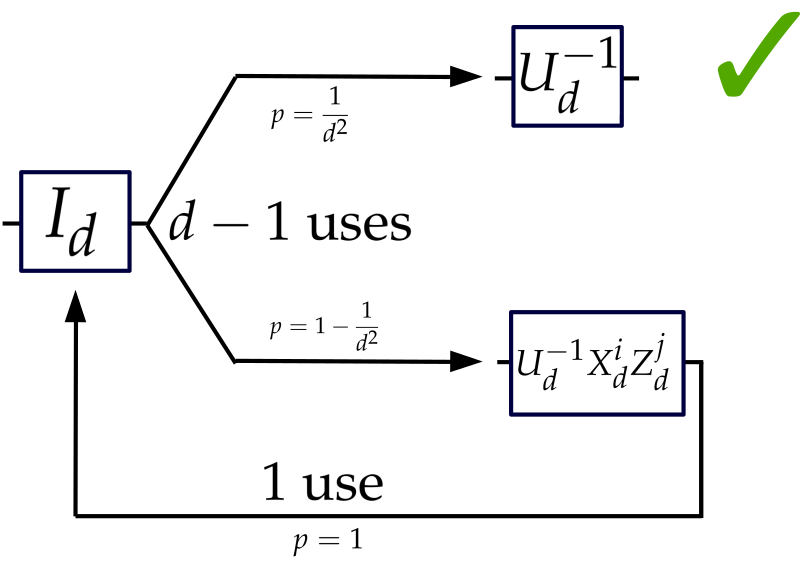} 
	\end{center}
			\caption{Flowchart illustrating the adaptive unitary inverse protocol.} \label{fig:inv_sequential}
	\end{figure}
\begin{enumerate}
\item We start by making a $d-1$ uses of the unitary $U_d$ to implement the probabilistic heralded transposition protocol for unitary inverse used in Theorem\,4 in the main text. When the generalised Bell measurements return the outcomes $i$ and $j$, the operation $U_d$ is transformed into $V_1=U_d^{-1} X^i_d Z^j_d$, where $X^i_d $ and $Z^j_d $ are the clock and shift operators, respectively (see Eq.\,\eqref{eq:shift}).
\item If both outcomes $i$ and $j$ correspond to the identity operator, \ie, $i=j=0$, we have $V_1=U_d^{-1}$ and we stop the protocol with success. If some other outcome is obtained, we make a single use of $\map{U_d}$ to apply ${X_d^i}^{-1}{Z_d^j}^{-1} U_d$ into $V_1$ and invert the transformation of step 1 to obtain identity operator
$\left[ {X_d^i}^{-1}{Z_d^j}^{-1} U_d   U_d^{-1}  Z_d^j X_d^i=I_d \right]$.
\item Go to step 1. 
\end{enumerate}

We see that step 1 requires $d-1$ uses and returns $\map{U^{-1}_d}$ with probability $\left(1-\frac{1}{d^2}\right)$. We need in total $d$ uses of $\map{U_d}$ to complete steps 1 and 2. Iteration of this protocol leads to a success probability of $p=1-\left(1-\frac{1}{d^2}\right)^{\floor{\frac{k+1}{d}}}$.

\subsection{Optimal protocols via SDP formulation and indefinite causal order advantage} \label{sec:SDP_inv}

We now apply the SDP methods of Sec.\,\ref{sec:SDP} to the case of unitary inversion and reproduce Table 1 of Ref.\,\cite{PRL} in in Table\,\ref{table:SDP_inv}. For qubits ($d=2$) {we note that, with the Pauli qubit unitary operator $Y$, $YU_2Y=U_2^*$ for every $U_2$ with determinant one. Hence, any protocol for transforming a single use of a qubit unitary operation into a single use of its transposition can be converted into a qubit unitary inversion protocol, \textit{vice versa}. Hence the results and conclusions for qubits are equivalent to the ones presented in Sec.\,\ref{sec:SDP_trans}.}

	For qutrits ($d=3$), Theorem \ref{theo:inv_nogo} ensures that circuits with a single use necessarily have success probability equal zero. For the case $d=3$ and $k=2$, parallel, adaptive, and indefinite causal order protocols have attained the same success probability, suggesting that our parallel unitary inversion protocol may be optimal when $d=k-1$.

\begin{table} 
\begin{tabular}{|c|c|c|c|}
\hline 
$d=2$ & Parallel & {Adaptive} & Indefinite causal order\\
\hline 
$k=1$ & \blue{$\frac{1}{4}=0.25$}  &\blue{$\frac{1}{4}=0.25$}  & \blue{$\frac{1}{4}=0.25$}   \\ 
\hline 
$k=2$ & \blue{$\frac{2}{5}=0.4$}  & {$0.4286\approx \frac{3}{7}$} & $0.4444\approx\frac{4}{9}$  \\ 
\hline 
$k=3$ & \blue{$\frac{1}{2}=0.5$}  & $0.7500\approx\frac{3}{4}$ & 0.9416  \\ 
\hline 
%
%
\hline 
$d=3$  & Parallel & {Adaptive} & Indefinite causal order\\
\hline 
$k=1$ & \blue{$0$}  & \phantom{aaaa} \blue{$0$} \phantom{aaaa}  & \blue{$0$}  \\ 
\hline 
$k=2$ & {$0.1111 \approx \frac{1}{9}$}  & $0.1111 \approx \frac{1}{9}$ & $0.1111 \approx \frac{1}{9}$  \\ 
\hline 
\end{tabular} 
\caption{Maximum success probabilities for {universally inverting} $k$ uses of $U_d$, by parallel quantum circuit, adaptive quantum circuit, and protocols with indefinite causal orders. Values in blue are analytical and in black via numerical SDP optimisation. This table is extracted from Table\,1 of Ref.\,\cite{PRL}.}\label{table:SDP_inv}
\end{table}

\section{Conclusions} \label{sec:conclusions}

		We have addressed the problem of designing probabilistic heralded universal quantum protocols that transform $k$ uses of an arbitrary (possibly unknown) $d$-dimensional unitary quantum operation $\map{U_d}$ to exactly implement a single use of some other operation given by $f(\map{U_d})$. For the cases where $f$ is a linear supermap, we have provided a SDP algorithm that  can be used to analyse parallel, adaptive, and indefinite causal order protocols. {For the parallel and adaptive cases, our algorithm finds a quantum circuit that universally implements the desired transformation with the optimal probability of success for any $k$ and $d$. For the indefinite causal order case, the algorithm finds a quantum process that obtains the desired transformation with the optimal probability of success for any $k$ and $d$.}
		
		For the particular case of unitary complex conjugation, \ie,  $f(\map{U_d})=\map{U_d^*}$ we have proved that when $k<d-1$ the success probability is necessarily zero, even when indefinite causal order protocols are considered. Since a deterministic parallel quantum circuit to transform $k=d-1$ uses of a general unitary operation $\map{U_d}$ { into a single use of its complex conjugation was presented in Ref.\,\cite{miyazaki17}, we can argue that the theoretical possibility of implementing universal exact unitary complex conjugation is completely solved.}
		
		For the particular case of unitary transposition, \ie $f(\map{U_d})=\map{U_d^T}$, we have shown that the optimal success probability with parallel circuits is exactly  ${p=1-\frac{d^2-1}{k+d^2-1}}$. When adaptive circuits are considered, we have presented an explicit protocol that has success probability ${p=1-\left(1-\frac{1}{d^2}\right)^{\ceil{\frac{k}{d}}}}$, which has an exponential improvement over any parallel protocol. We have also shown that indefinite causal order protocols outperforms causally ordered ones by tackling the cases $d=2,k\leq 3$ and $d=3,k=2$ numerically.
		
For the particular case of unitary inversion, \ie $f(\map{U_d})=\map{U_d^{-1}}$, we have proved that when $k<d-1$ the success probability is necessarily zero, even when indefinite causal order protocols are considered. When $k\geq d-1$ we have presented parallel and adaptive circuits to succeed in this task and proved that the success probability of our adaptive protocol presented in Ref.\,\cite{PRL} has probability of success given by ${p=1-\left(1-\frac{1}{d^2}\right)^{\floor{\frac{k+1}{d}}}}$ and we prove it to be exponentially higher than any success probability obtained by parallel circuits.

\textit{Acknowledgements: --} %
 We are indebted to A. Abbott, M. Araújo, J. Bavaresco, D. Gross, and P. Guerin for valuable discussions.

{This work was supported by MEXT Quantum Leap Flagship Program (MEXT Q-LEAP) Grant Number JPMXS0118069605, Japan Society for the Promotion of Science (JSPS) by KAKENHI grant No. 15H01677, 16F16769, 16H01050, 17H01694, 18H04286, 18K13467 and Advanced Leading Graduate
Course for Photon Science (ALPS).}

\providecommand{\href}[2]{#2}\begingroup\raggedright\endgroup

\end{document}